 \documentclass[final,3p,times]{elsarticle}

\usepackage[english]{babel}
\usepackage[latin1]{inputenc}
\usepackage{amsfonts} 
\usepackage{amsmath}
\usepackage{amssymb}
\usepackage{amsthm}
\usepackage{enumerate}
\usepackage{hyperref} 
\usepackage{graphicx} 
\usepackage[usenames,dvipsnames]{color}
\usepackage{rotating}
\usepackage{epsfig} 
\usepackage{graphicx}
\usepackage{threeparttable}
\usepackage{multirow}
\usepackage{url}
\usepackage{booktabs} 
\usepackage{lineno}

\newtheorem{definition}{Definition}[section]

\newtheorem{theorem}[definition]{Theorem}
\newtheorem{lemma}[definition]{Lemma}

\journal{Journal  }

\begin{document}

\begin{frontmatter}


\title{A new approach for sizing trials with composite binary endpoints using anticipated marginal values and accounting for the correlation between components}



\author{Marta Bofill Roig}
\ead{marta.bofill.roig@upc.edu}
\author{Guadalupe G\'omez Melis}  
\ead{lupe.gomez@upc.edu}

\address{Departament d'Estad\'{i}stica i Investigaci\'{o} Operativa, Universitat Polit\`{e}cnica  de Catalunya, Barcelona, Spain}

\begin{abstract}
Composite binary endpoints are increasingly used as primary endpoints in clinical trials. When designing a trial, it is crucial to determine the appropriate sample size for testing the statistical differences between treatment groups for the primary endpoint. As shown in this work, when using a composite binary endpoint to size a trial, one needs to specify the event rates and the effect sizes of the composite components as well as the correlation between them. In practice, the marginal parameters of the components can be obtained from previous studies or pilot trials, 
however, 
the correlation is often not previously reported and thus usually unknown. We first show that the sample size for composite binary endpoints is strongly dependent on the correlation and, second, that slight deviations in  the prior information on the marginal parameters may result in underpowered trials for achieving the study objectives at a pre-specified significance level.

We propose a general strategy for calculating the required sample size when the correlation is not specified, 
and accounting 
for uncertainty in the marginal parameter values. We present the web platform CompARE to characterize composite endpoints and to calculate the sample size just as we propose in this paper. We evaluate the performance of the proposal with a simulation study, and illustrate it by means of a real case study using CompARE.
\end{abstract}


\end{frontmatter}




\section{Introduction}\label{sec1}

Many trials are designed to evaluate more than one endpoint with the aim of providing a wider picture of the intervention effects \cite{FDA2017,Rosenblatt2017}. When the rate of occurrence of an event is expected to be low, it is common to consider the composite event   defined as the occurrence of any of a set of pre-specified events. 
This composite event is usually chosen as the primary efficacy endpoint for comparing two treatment groups, either by comparing proportions between groups at the end of the study or by using time-to-event analysis. In this paper, we focus on composite binary endpoints.

Power analysis  and its subsequent sample size calculation have been widely discussed in the literature on comparing two proportions in the univariate case \cite{Lachin1981,Donner1984,Fleiss1981,Friedman1981}. These standard sample size formulae are based on the effect size and the frequency of occurrence of primary endpoint, and they could be applied in a straightforward way to a composite endpoint if its effect size and frequency are known prior to the initiation of the study. However, the effect size and frequency of observing the composite endpoint depend on the corresponding effect and frequency of the composite components, which are often quite dissimilar and thus make the composite parameters  very difficult to anticipate.

The TACTICS-TIMI $18$ trial \cite{Cannon2001}   illustrates some problems that might arise when determining the sample size for a primary composite binary endpoint. TACTICS-TIMI 18 was an international, multicenter, randomized trial that evaluated the efficacy of invasive and conservative treatment strategies in patients with unstable angina or non-Q-wave acute myocardial infarction treated with tirofiban, heparin, and aspirin. The primary hypothesis of the TACTICS-TIMI 18 trial was that an early invasive strategy would reduce the combined incidence of death, acute myocardial infarction, and rehospitalization for acute coronary syndromes at six months when compared with an early conservative strategy. The primary endpoint was the composite endpoint formed by   death,    non-fatal myocardial infarction, and   rehospitalization for acute coronary syndrome at 6 months.

Similar research questions such as those in TACTICS-TIMI 18 were previously investigated in the TIMI IIIB and VANQWISH trials\cite{Cannon1998}. The TIMI IIIB trial \cite{TIMIIIIB} considered the primary composite endpoint of death, post-randomization myocardial infarction, and a positive exercise test at 6 weeks; whereas the primary endpoint in the VANQWISH trial \cite{VANQWISH}   was the combination of death and  non-fatal myocardial infarction at 12 months of follow-up.
The initial planning of TACTICS-TIMI $18$ was based on those trials expecting $22\%$ events  of the primary composite endpoint in the conservative-strategy group,  to detect a relative difference of $25\%$ between the two groups for a $80\%$ power. Those anticipated values resulted in the need to recruit at least $ 1720 $ patients. However, TACTICS-TIMI $ 18 $ yielded a $19\%$ frequency of observing the combination of death, acute myocardial infarction and rehospitalization at six months, which was remarkably lower than expected and delivered a relative difference of $20\%$ between groups, a figure that is seriously lower than the anticipated $25\%$. Note that if the anticipated frequency of observing the composite endpoint had been closer to the observed results, at least $ 2000 $ patients rather than $ 1720 $ would have been required and the sample size needed would have been larger than the one initially planned.

In this paper, we present sample size formulations for detecting a hypothesized difference between treatments in a primary composite binary endpoint based on the event rates and effect sizes of the composite components. The motivation for this is mainly because prior information on the marginal effects and event rates is commonly available from previous or pivotal studies, as illustrated in the TACTICS-TIMI $ 18 $ trial. Moreover, the major findings in a trial with a primary composite endpoint should be well supported by its components \cite{FDA2017, EMA2016}, since the trial could be considered negative if the components are not in line with the result \cite{Pocock2015, ICH9}. Nevertheless, as shown in this paper, the sample size calculation for composite endpoints relies not only on the anticipation of the effect size and the event rates of the composite components, but also on the correlation between them. However, even though the marginal parameters could be obtained previously, the correlation is usually not reported in practice and, thus, is frequently unknown and difficult to anticipate.

Several authors have addressed the correlation's influence on sample size determination when more than one endpoint is used as the primary endpoint. Sozu et al. \cite{Sozu2010} discuss several methods for calculating power and sample size for multiple co-primary binary endpoints, and they study the impact on the sample size, specifically regarding the association among endpoints. Senn and Bretz\cite{Senn2007} examine sample size for trials under different power definitions for multiple testing problems.
Rauch and Kieser \cite{Rauch2012} and Sander et al. \cite{Sander2016} define a multiple test procedure focused on a composite binary endpoint and a pre-specified main component, and propose an internal pilot study for estimating the unknown parameters and revising the sample size.
However, to the best of our knowledge, methodologies are limited in regard to handling the sample size calculation for composite binary endpoints when the correlation is unknown.

The aim of this paper is two-fold. First, we focus on providing a general procedure for sizing trials with composite binary endpoints, doing so on the basis of anticipated information of the composite components even if the correlation is unknown. We show that the sample size for composite binary endpoints is strongly dependent on the correlation, and that slight deviations in the prior information on the marginal parameters may result in trials being too underpowered for achieving the study objectives at the pre-specified significance level. We propose a sample size strategy to calculate the minimum sample size that guarantees the planned power while accounting for, on the one hand, the uncertainty of the correlation value and, on the other, plausible deviations in the marginal parameter values. Second, we present CompARE, a freely available web-based tool for characterizing binary composite endpoints and computing the needed sample size under several settings. CompARE provides aids to help understand the role played by each one of the components of the composite endpoint, as well as their consequences on the required sample size. The methodologies presented in this paper have been implemented in CompARE.

This paper is structured as follows. 
In Section \ref{Section.Notation}, we introduce the settings of the problem and the CompARE web tool. In Section \ref{Section.SampleSize}, we review sample size planning when evaluating risk difference.  In Section \ref{Section.ConvenientSampleSize}, we present sample size formulae for composite binary endpoints based on the parameters of the components plus the correlation. We further describe the performance of these formulae according to the parameters and propose a strategy for sizing trials when the correlation is unknown.
In Section \ref{Section.casestudy},   we exemplify the proposal by the TACTICS-TIMI $ 18 $ trial using CompARE, and in Section \ref{Section.extension} we extend the proposal to those trials for which the treatment effect is measured by the relative risk or odds ratio.  In Section \ref{Section.simulations}, we investigate the performance of the power and significance level under misspecification of the correlation and evaluate the proposed sample size strategy with a simulation study. We conclude the paper with the Discussion.


\section{Notation, assumptions and CompARE} \label{Section.Notation}

We consider a randomized clinical trial comparing two treatment groups: the control group ($ i=0 $) and treatment group ($ i=1 $), each  one composed of $n^{(i)}$ patients who are followed for a pre-specified time $\tau$. 
For simplicity, we consider only two events of potential interest, $\varepsilon_1$ and $\varepsilon_2$. Let $X_{ijk}$ denote the response of the $k$-th binary endpoint for the $j$-th patient in the $i$-th group of treatment ($ i=0,1 $, $j=1,...,n^{(i)}$, $k=1,2$).
The response  $X_{ijk}$ is defined by $1$ if the event, $\varepsilon_k$, has occurred during the follow-up  and $0$ otherwise.   

We define the binary  composite endpoint  as the event that occurs whenever one of the endpoints is observed, that is, $\varepsilon_* = \varepsilon_1 \cup \varepsilon_2$.   At this point we assume that the composite endpoint is well-defined, that is, both composite components are important enough to be considered; and we include those adverse clinical outcomes that are relevant to the clinical setting. 
We denote by $X_{ij*}$ the composite response defined as a Bernoulli
random variable with  probability of observing the event \\
$p^{(i)}_{*}= \mathrm{P}(X_{ij*} =1) =1-
q_*^{(i)}  $, where:
\begin{eqnarray} \label{def.CBE} 
	X_{ij*} &=&
	\begin{cases}
		1, \text{ if } X_{ij1} + X_{ij2} \geq 1 \\
		0, \text{ if else } X_{ij1} + X_{ij2} = 0
	\end{cases}
\end{eqnarray}

To evaluate whether there is  a risk reduction in the treatment group compared with the control group, we set a hypothesis test where the null hypothesis states that there is no difference between the control and the treatment groups; whereas the alternative hypothesis assumes a risk reduction  in the treatment group.  
The usual measures to evaluate the treatment effect when comparing two groups are   
the difference in proportions (also called risk difference), denoted by $\delta_*$; the relative risk (or risk ratio), $\textrm{R}_*$; and the  odds ratio, $\textrm{OR}_*$.  
The relationship between these measures and the probabilities of observing the binary composite endpoint in each group are given in Table \ref{Table.TreatmentMeasures}, together with the null and alternative hypothesis that should be set in each case.
The following sections will be developed in terms of the risk difference $\delta_* = p_*^{(1)} - p_*^{(0)}$ of the composite binary endpoint.  Section \ref{Section.extension} extends the results to the relative risk and odds ratio.

\begin{table}[h!] 
	\centering
	\caption[position=above]{Parameter to anticipate the effect, and set of hypotheses. }
	\begin{tabular}{cccc}
		\toprule 
		&  \textbf{Parameter Effect} & \textbf{Null hypothesis} & \textbf{Alternative hypothesis}   \\
		\toprule  
		Risk difference & $\delta_* = p_*^{(1)} - p_*^{(0)}$  & $\delta_* = 0$  & $\delta_* < 0 $  \\   
		Relative risk & $\textrm{R}_* = p_*^{(1)}/p_*^{(0)}$  & $ \log(\textrm{R}_*) = 0 $ & $ \log(\textrm{R}_*) < 0 $ \\ 
		Odds ratio & $\textrm{OR}_* = \frac{p_*^{(1)}/q_*^{(1)}}{p_*^{(0)}/q_*^{(0)}}$  & $ \log(\textrm{OR}_*) = 0 $ & $ \log(\textrm{OR}_*) < 0 $ \\ 
		\bottomrule
	\end{tabular} 
	\label{Table.TreatmentMeasures}
\end{table}

\subsection{An insight into the parameters of the composite endpoint}  

Let  $p_k^{(i)}$ and $q_k^{(i)}$ represent the probabilities that $\varepsilon_k$ occurs or not, respectively, for a  patient belonging to the $i$-th group. Let $\rho^{(i)}$ denote Pearson's correlation coefficient between the components in the $i$-th group.
The probability of observing the composite event $\varepsilon_*$ is in terms of the probabilities of $\varepsilon_1$ and $\varepsilon_2$ and the correlation, as follows:
\begin{eqnarray} \label{probCBE}
	p^{(i)}_{*} &=& 1-  q_1^{(i)} q_2^{(i)} - \rho^{(i)}   \sqrt{p_1^{(i)} p_2^{(i)} q_1^{(i)} q_2^{(i)}} \ , \hspace{8mm} i=0,1
\end{eqnarray}
Note here that the probability of observing the composite endpoint becomes smaller as the correlation between the components of the composite increases. 

The effect size in the composite endpoint  in terms  of the risk difference,  $\delta_*$, is given by: 
\begin{eqnarray} \label{effectCBE.riskdiff}
	\delta_* = \delta_1 q_2^{(0)} + \delta_2 q_1^{(0)}  - \delta_1 \delta_2 
	+ \rho^{(0)} \sqrt{p_1^{(0)}p_2^{(0)}q_1^{(0)}q_2^{(0)}} - \rho^{(1)}\sqrt{(p_1^{(0)}+\delta_1)(p_2^{(0)}+\delta_2)(q_1^{(0)}-\delta_1)(q_2^{(0)}-\delta_2)}  
\end{eqnarray}
where  $\delta_k$ ($k=1,2$) corresponds to the  risk
difference for  each of its components.

From now on the correlation is assumed equal for both groups and denoted by $\rho$, that is, $\rho=\rho^{(0)}=\rho^{(1)}$.
Let $\theta$ denote the vector of  event rates of the composite components in the control group, that is, $ \theta= (p_1^{(0)}, p_2^{(0)})$, and let $\lambda$ represent the vector of marginal effect sizes, that is, $\lambda=(\delta_1, \delta_2)$. We will denote the risk difference as a function of the marginal parameters  ($\theta,\lambda$)  and the correlation $\rho$ by $\delta_*(\theta,\lambda,\rho)$; and the probability of observing $\varepsilon_*$ under the control group by $p^{(0)}_{*}(\theta, \rho)$. 
We remark here that	when  $\lambda$ and $\theta$ are fixed such that $p_k^{(0)}<0.5$ and $\delta_k<0$ ($k=1,2$), 
the risk difference $\delta_*(\theta,\lambda,\rho)$ increases with respect to the correlation $\rho$ (see Appendix \ref{app.treateffect}).

\subsection{CompARE} \label{Section.CompARE} 

We present CompARE\footnote{Link to CompARE: \url{https://cinna.upc.edu/compare/}, the open-source code for CompARE is available at: \url{https://github.com/MartaBofillRoig/CompARE}}, an open-source and completely free web platform that can be used as a tool for clinicians, medical researchers and statisticians to compute the sample size according to the procedure proposed in this paper. Furthermore, CompARE can be used to:
\begin{enumerate}
	\item Determine the sample size for different situations, among them, when the correlation is not known.
	\item Specify the treatment effect for the composite endpoint based on the marginal information of the composite components, and to study the performance of the composite parameters according to them.
	\item Calculate and interpret the measures of association among the composite components, then investigate their characteristics.
	\item Choose the best primary endpoint to lead the trial. CompARE computes the Asymptotic Relative Efficiency method\cite{GomezLagakos, BofillGomez}, which quantifies differences in the efficiency of using -- as the primary endpoint -- a composite endpoint over one of its components.
\end{enumerate} 
Figure \ref{fig:compare}  summarizes all the capabilities of CompARE. To use CompARE, the least you should provide is the effect size and event rates of the composite components as well as the correlation.

\begin{figure}[h!]
	\centering
	\includegraphics[width=1\linewidth]{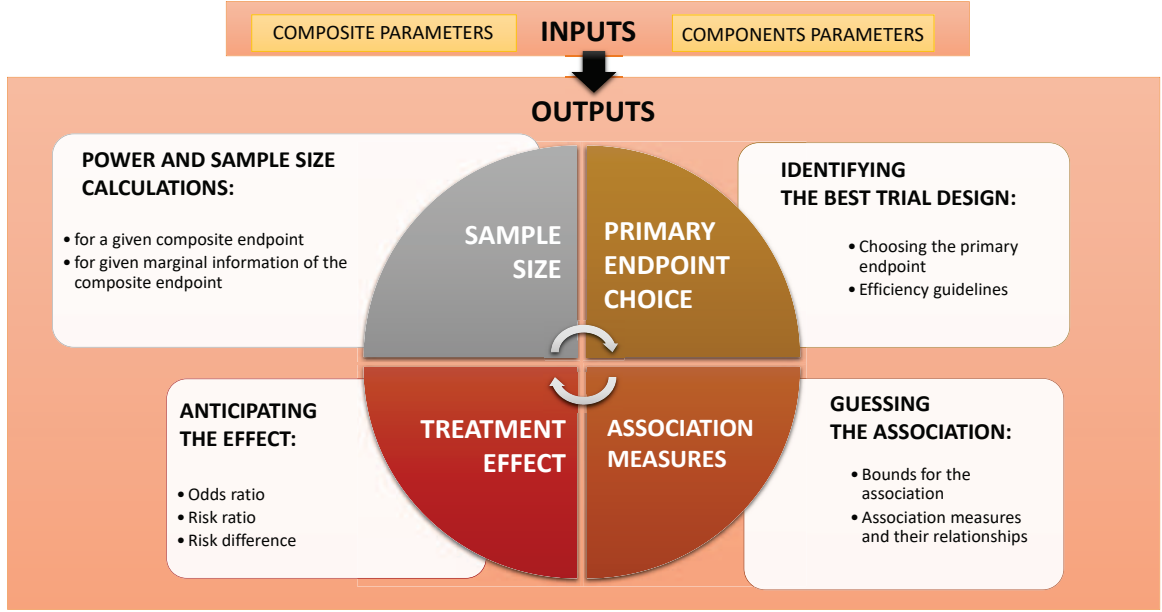}
	\caption{Scheme of the main lines of action in CompARE. Link to the CompARE homepage: \url{http://cinna.upc.edu/compare/}. Open-source code is available at: \url{https://github.com/MartaBofillRoig/CompARE}.}
	\label{fig:compare}
\end{figure}


\section{Sample Size when the parameters of the composite endpoint  can be anticipated} \label{Section.SampleSize}

In this section we summarize the statistics and sample size formulae to test for a risk difference when the probability of occurrence in the control group  of the composite binary endpoint can be  anticipated and  for a given expected 
risk difference. Since the composite endpoint is  an univariate outcome, a single statistical test is performed and, consequently, no multiplicity problem occurs and no statistical adjustment is needed. Therefore, as we will see,  the formulas follow the univariate case and are straightforward but to make the paper comprehensive and the following sections meaningful, we displayed them in terms of the composite endpoint parameters.

Herein we assume a clinical trial  where, first, patients are randomized to one of two treatment arms following a balanced design  and, second,   where the  primary endpoint is a  binary composite endpoint. The aim is to detect a hypothesized  risk reduction in the primary composite endpoint at the  significance level of $\alpha$ and with desired power  equal to $1-\beta$. Let $n$ be the total sample size  required, with $n^{(i)}=n\big/2$ patients per group ($i=0,1$);  and let us denote by    $z_\alpha$ and $z_\beta$   the values of standardized normal deviates corresponding to  $\alpha$ and $\beta$.

The null hypothesis  is stated as $\textrm{H}_0^*:  p_*^{(1)}-
p_*^{(0)}=0$ and  is compared against the alternative hypothesis $\textrm{H}_1^*:
p_*^{(1)}- p_*^{(0)}<0$.  
To test  $\textrm{H}_0^*$ against  $\textrm{H}_1^*$ we use the statistic:
\begin{eqnarray} \label{test.difprop.pooled}
	\textrm{T}_{*,n} &=&  \frac{\widehat{p}_*^{(1)} - \widehat{p}_*^{(0)}}{\sqrt{\widehat{Var}(\widehat{p}_*^{(1)} - \widehat{p}_*^{(0)})}} 
\end{eqnarray}
where $\widehat{p}_*^{(i)} =  \frac{1}{n^{(i)}} \sum_{j=1}^{n^{(i)}} X_{ij*}$. Under $\textrm{H}_0^*$, $\textrm{T}_{*,n}$ follows, asymptotically, the standard normal distribution. We will  reject the null hypothesis at the $\alpha$ level of significance if $\textrm{T}_{*,n} <- z_\alpha$.
The   variance $Var(\widehat{p}_*^{(1)} - \widehat{p}_*^{(0)})$  in equation \eqref{test.difprop.pooled} can be estimated  under  $\textrm{H}_0^*$ using the    pooled variance estimate\cite{Donner1984}:
$$\widehat{Var}_{H_0}(\widehat{p}_*^{(1)} - \widehat{p}_*^{(0)}) =  
\frac{1 }{2n^{(0)}} \cdot\left(  \widehat{p}_*^{(0)} + \widehat{p}_*^{(1)} \right) \cdot \left(  \widehat{q}_*^{(0)} + \widehat{q}_*^{(1)} \right)    $$
or under  $\textrm{H}_1^*$ using the unpooled variance estimate: 
$$\widehat{Var}_{H_1}(\widehat{p}_*^{(1)} - \widehat{p}_*^{(0)}) =   \frac{1}{n^{(0)}  }  \left( \widehat{p}_*^{(0)} \widehat{q}_*^{(0)} + \widehat{p}_*^{(1)} \widehat{q}_*^{(1)} \right)   $$ 

For a given probability under control group $p_*^{(0)}$, the  required sample size  using the pooled estimate to have power $1-\beta$ in order to detect an effect size of  $\delta_*$ at a significance level $\alpha$ is given by \cite{Lachin1981,Fleiss1981}: 
\begin{eqnarray} \label{samplesize.p.pooled} 
	n &=&  2 \cdot \left(  z_\alpha \cdot \sqrt{ (2p_*^{(0)} + \delta_* )  (2q_*^{(0)} - \delta_* ) } + z_\beta \cdot \sqrt{p_*^{(0)}q_*^{(0)} + (p_*^{(0)}+\delta_*) (q_*^{(0)}-\delta_*) }   \right)^2 \Bigg/ \delta_*^{2}
\end{eqnarray} 
Note that in \eqref{samplesize.p.pooled} 
we have replaced $p_*^{(1)}$ with  $p_*^{(0)}+\delta_*$.  

Similarly,  the corresponding sample size  using the unpooled variance estimate is given by: 
\begin{eqnarray}  \label{samplesize.p.unpooled} 
	n  &=&  2 \cdot  \left( \frac{z_\alpha + z_\beta}{\delta_*} \right)^2 \cdot \left( p_*^{(0)} q_*^{(0)} + (p_*^{(0)}+\delta_*) (q_*^{(0)}-\delta_*)\right) 	
\end{eqnarray}
Note that, under the null hypothesis $\textrm{H}_0^*:  p_*^{(1)}-p_*^{(0)}= 0$, expressions   \eqref{samplesize.p.pooled} and \eqref{samplesize.p.unpooled} coincide.


\section{Sample Size   based on anticipated values of the composite components} \label{Section.ConvenientSampleSize}

Sample size formulae underlined in Section \ref{Section.SampleSize} are based on the parameters of the composite endpoint, that is, the event rate under the control group, $p_*^{(0)}$, and the treatment effect, $\delta_*$. In this section, we derive  the sample size based on the anticipated information on the marginal parameter values and the correlation, even if the correlation value 
is not fully specified and/or the event rates values are not accurately anticipated.

\subsection{Sample size based on composite components}

Given the event rates in the control group $ \theta= (p_1^{(0)}, p_2^{(0)})$, the expected effect size for each component $\lambda=(\delta_1,\delta_2)$, and the correlation  between the occurrence of both components $\rho$, we will denote by $n(\theta, \lambda, \rho)$ the needed sample size, which is computed by using the unpooled variance estimate, to detect a risk difference $\delta_*(\theta, \lambda, \rho)$ (see equation \eqref{effectCBE.riskdiff}) at significance level $\alpha$ with $1-\beta$ power.

The expression for $n(\theta, \lambda, \rho)$ is obtained after direct substitution into formula \eqref{samplesize.p.unpooled} and is as follows:
\begin{eqnarray} \label{nCBE}  
	n(\theta, \lambda, \rho) &=& \frac{2 \cdot \left( z_\alpha + z_\beta \right)^2 \cdot \left( p_*^{(0)}(\theta,\rho)\left(  1-p_*^{(0)}(\theta,\rho)\right)  +  \left(p_*^{(0)}(\theta,\rho) + \delta_*(\theta, \lambda, \rho) \right) \left(1-p_*^{(0)}(\theta,\rho) - \delta_*(\theta, \lambda, \rho)\right)  \right) }{\delta_*(\theta, \lambda, \rho)^2}    	
\end{eqnarray}
where $p_*^{(0)}(\theta,\rho)$ is given in \eqref{probCBE}.
Note that the sample size also relies on the  significance level $\alpha$ and the power $1-\beta$, but these are omitted for ease of notation. 
The corresponding sample size under the pooled estimate can be analogously calculated by using  $\theta$, $\lambda$ and $ \rho $ and its expression can be found in the online   support material.

\subsection{Sample size bounds } \label{section.SSbehavior}

Assuming that the correlation is the same in the two treatment groups,  it follows that the correlation takes values between the lower bound, $B_L(\cdot)$, and the upper bound, $B_U(\cdot)$, which are functions of $\theta$ and $\lambda$, and are defined as:
\begin{eqnarray} \label{corr.bounds}
	B_L(\theta,\lambda)
	&=& \max \left\lbrace
	- \sqrt{\frac{p_1^{(0)} \cdot p_2^{(0)}}{ q_1^{(0)} \cdot q_2^{(0)}}}, \ \
	- \sqrt{\frac{q_1^{(0)} \cdot  q_2^{(0)}}{ p_1^{(0)} \cdot  p_2^{(0)}}},  \ \
	- \sqrt{\frac{(p_1^{(0)}+\delta_1) \cdot (p_2^{(0)}+\delta_2)}{ (q_1^{(0)}-\delta_1) \cdot (q_2^{(0)}-\delta_2)}}, \ \
	- \sqrt{\frac{(q_1^{(0)}-\delta_1) \cdot  (q_2^{(0)}-\delta_2)}{ (p_1^{(0)}+\delta_1) \cdot  (p_2^{(0)}+\delta_2)}} \right\rbrace\\ \nonumber
	B_U (\theta,\lambda)&=&
	\min \left\lbrace
	+\sqrt{\frac{p_1^{(0)} \cdot q_2^{(0)}}{ p_2^{(0)} \cdot q_1^{(0)}}}, \ \
	+\sqrt{\frac{p_2^{(0)} \cdot q_1^{(0)}}{ p_1^{(0)} \cdot q_2^{(0)}}}, \ \
	+\sqrt{\frac{(p_1^{(0)}+\delta_1) \cdot (q_2^{(0)}-\delta_2)}{ (p_2^{(0)}+\delta_2) \cdot (q_1^{(0)}-\delta_1)}}, \ \
	+\sqrt{\frac{(p_2^{(0)}+\delta_2) \cdot (q_1^{(0)}-\delta_1)}{ (p_1^{(0)}+\delta_1) \cdot (q_2^{(0)}-\delta_2)}} \right\rbrace
\end{eqnarray}
Note that when at least one of the event rates is very close to $0$, the lower bound $B_L(\lambda,\theta)$ will also be close to $0$ and the plausible correlation values will be always positive.
We also notice that, in clinical trials the  probabilities of observing the events are often quite low and commonly smaller than $0.5$. In this case,  the expressions for  $B_L(\lambda, \theta)$ and $B_U(\lambda, \theta)$ can be simplified.  See the online supplementary material for more details.

Considering such bounds for a given marginal parameters $\theta$ and $\lambda$, the sample size $n(\theta, \lambda, \rho)$ is an increasing function of the correlation $\rho$, 
and it is bounded below and above by $ n(\theta,\lambda,B_L(\theta,\lambda)) $ and $ n( \theta,  \lambda, B_U(\theta,\lambda)) $, respectively. As a consequence, the more correlated the single endpoints are, the larger will be the necessary sample size for detecting the differences between groups in the composite endpoint. Details for this derivation  are provided in Appendix  \ref{app.samplesizeCBE}  (see   Theorem 1).

\subsection{Sample size  with uncertain correlation value} \label{SSmethod}

Since the correlation plays an important role in calculating the sample size, we propose a strategy for deriving the sample size when the parameters that correspond to the composite components are known and the correlation value is not specified in advance.

Prior knowledge about the effect of the treatment being investigated can lead to scientists foreseeing whether the two events of interest,   $\varepsilon_1$ and $\varepsilon_2$,  are weakly, moderately or strongly correlated.  We allow for prior information by splitting the rank of the correlation into three equal-sized intervals, and we consider three correlations categories: weak for the interval whose correlation values are lower; moderate for those intermediate correlation values; and strong for those correlation values that are higher. If any information exists, we will take it into account and will proceed as follows:

\begin{enumerate}[(i)]
	\item  \textit{Correlation bounds for each category:} \\
	Considering the categories weak/moderate/strong for the correlation, the  plausible correlation values for a given ($ \theta,\lambda $)  are in this situation those between the lower  and   upper   values within each category. 	If the events are weakly correlated, the    correlation is  between  $B_L(\theta,\lambda)$ and $ \left(B_U(\theta,\lambda)-B_L(\theta,\lambda) \right) \big/3$;  if they are  moderately correlated, its value lies between $ \left(B_U(\theta,\lambda)-B_L(\theta,\lambda) \right) \big/3$ and $ 2\cdot\left(B_U(\theta,\lambda)-B_L(\theta,\lambda) \right) \big/3$; and if they are  strongly correlated, it  is between $ 2\cdot\left(B_U(\theta,\lambda)-B_L(\theta,\lambda) \right) \big/3$  and $ B_U(\theta,\lambda) $.\\
	If we cannot place the correlation in any of the above categories, 
	we use the most severe case within  its plausible values, then, $B_U(\theta,\lambda)$. (See Table \ref{Table.SampleSizeApproach2}).
	\item \textit{Calculate the sample size in each category:} \\
	For the sample size, we advocate using the maximum sample size across all its possible values. 
	That is, $n(\theta,\lambda,\left(B_U(\theta,\lambda)-B_L(\theta,\lambda) \right) \big/3)$, $n(\theta,\lambda,2\left(B_U(\theta,\lambda)-B_L(\theta,\lambda) \right) \big/3)$, and $n(\theta,\lambda,B_U(\theta,\lambda)) $ for  weak, moderate or strong correlations, respectively. Note that since  we are assuming the correlation value  that maximizes the sample size across its plausible values, we are guaranteeing that the pre-specified power  $1-\beta$ is attained.\\
	If the correlation value can not be ascribed to any category,
	then, we propose a conservative sample size strategy of using the overall possible    maximum   sample size, that is, $n(\theta,\lambda,B_U(\theta,\lambda))$. 
	Table \ref{Table.SampleSizeApproach2} outlines 
	the range of  correlations and sample sizes values, together with  the proposed sample size  for each category. 
\end{enumerate} 

{
	\begin{table}[h!] 
		\small
		\centering
		\caption[position=above]{ 
			Correlation category and its subsequent   correlation bounds,  $ B_L(\cdot) $ and $ B_U(\cdot) $ (given in    \eqref{corr.bounds})
			for event rates of the composite components $\theta=(p_1^{(0)},p_2^{(0)})$, and marginal effect sizes $\lambda=(\delta_1,\delta_2)$. Sample size bounds   for each correlation category and proposed sample size strategy calculated by \eqref{nCBE}   according to the margins ($\theta,\lambda$) and for given significance level $\alpha$ and power $1-\beta$. }
		\begin{tabular}{cccc}
			\toprule 
			Category &  Correlation Bounds & Sample Size Bounds &    Sample Size     \\
			\toprule     
			Weak & $  \left[  B_L(\theta,\lambda), \frac{(B_U(\theta,\lambda)-B_L(\theta,\lambda))}{3} \right)  $ & 
			$\Big[ n(\theta,\lambda,B_L(\theta,\lambda)) , \ \ n\left( \theta,\lambda,\frac{(B_U(\theta,\lambda)-B_L(\theta,\lambda))}{3} \right)  \Big]$ &  
			$ n\left( \theta,\lambda,\frac{(B_U(\theta,\lambda)-B_L(\theta,\lambda))}{3} \right)$ \\ [4mm]   
			Moderate & 
			$ \left[   \frac{(B_U(\theta,\lambda)-B_L(\theta,\lambda))}{3},  \frac{2(B_U(\theta,\lambda)-B_L(\theta,\lambda))}{3}  \right)  $ &
			$ \Big[  n\left(\theta,\lambda, \frac{(B_U(\theta,\lambda)-B_L(\theta,\lambda))}{3} \right) , \ \  n\left( \theta,\lambda,\frac{2(B_U(\theta,\lambda)-B_L(\theta,\lambda))}{3} \right)\Big] $  & 
			$ n\left(  \theta,\lambda,\frac{2(B_U(\theta,\lambda)-B_L(\theta,\lambda))}{3} \right) $ \\ [4mm]      
			Strong & 
			$  \left[   \frac{2(B_U(\theta,\lambda)-B_L(\theta,\lambda))}{3}, B_U(\theta,\lambda)  \right)  $ & 
			$ \Big[  n\left( \theta,\lambda,\frac{2(B_U(\theta,\lambda)-B_L(\theta,\lambda))}{3} \right) , \ \ n( \theta,\lambda,  B_U(\theta,\lambda)) \Big]$   & 
			$n(\theta,\lambda,B_U(\theta,\lambda))$  \\
			[4mm]
			No prior information  & $  \left[  B_L(\theta,\lambda), B_U(\theta,\lambda) \right)  $ & $\Big[ n(\theta,\lambda,B_L(\theta,\lambda)) , \ \ n\left( \theta,\lambda, B_U(\theta,\lambda) \right)  \Big]$ &  $ n\left( \theta,\lambda,B_U(\theta,\lambda) \right)$ \\
			\bottomrule
		\end{tabular} 
		\label{Table.SampleSizeApproach2}
\end{table} }

\subsection{Sample size accounting  for departures from the anticipated event rates} \label{Section.Bands}

The marginal parameters are often estimated through previous studies or pivotal trials with a limited number of patients and whose patient populations or concomitant drugs could differ from the current ones. Because of that, there is great uncertainty in the values that need to be anticipated for computing the sample size. In this section, we consider that the event rates   $p_1^{(0)}$ and $p_2^{(0)}$  have been previously  estimated  and their corresponding standard errors  of the point estimate are
provided.  

Let $I_k= \left[ \underline{p}_k^{(0)},\bar{p}_k^{(0)} \right]$ denote a set of plausible values  for the true value  of $p_k^{(0)}$. 
For instance, for those previous trials in which we have the standard deviations for the event rates, we can use the set of plausible values for $p_k^{(0)}$ that a $95\%$ confidence interval would yield.
We   address the issue of sizing a trial for a significance level $\alpha$ and power $1-\beta$
based on the intervals $I_1$ and $I_2$,
and for fixed effects $\delta_1$ and $\delta_2$ when the correlation value is not known.

We state that,  for given $\delta_1$ and $\delta_2$  and at fixed $\rho=r$, the   sample size $n(p_1^{(0)},p_2^{(0)},\lambda,r)$  (see equation \eqref{nCBE}) that is needed for power $1-\beta$ at a significance level $\alpha$, 
falls into the interval:
\begin{eqnarray} \label{CI_ss}
	 {\cal I}(r, I_1, I_2, \lambda) =
	[ \ n(\underline{p}_1^{(0)}, \underline{p}_2^{(0)}, \lambda, r), 
	\ \  n(\bar{p}_1^{(0)}, \bar{p}_2^{(0)}, \lambda, r)  \ ]
\end{eqnarray}

This interval is such that it contains the    sample size   required
to attain power $1-\beta$, which is necessary
for detecting an effect size equal to  $\delta_* = p_*^{(1)}- p_*^{(0)}$
at a significance level $\alpha$
according to the marginal effects $\delta_1$ and $\delta_2$,  the correlation $r$,
and the
event rates $p_k^{(0)}$ within $I_k$  ($k=1,2$). Note that the interval ${\cal I}(r, I_1, I_2, \lambda)$ gives us the plausible sample size values by taking into account the uncertainty of the marginal parameter values, and it provides us the maximum sample size that we would need even though the anticipated event rates are not accurate.

Considering $\Theta = (I_1, I_2,\lambda)$   the set of values for the marginal parameters, and denoting by $\rho_L(\Theta) = \max_{ (\pi_1, \pi_2) \in I_1 \times I_2} B_L(\pi_1, \pi_2, \lambda)$   and
$\rho_U(\Theta) = \min_{(\pi_1, \pi_2) \in I_1 \times I_2} B_U(\pi_1, \pi_2, \lambda)$   the lower and upper bounds of the correlation within the set $\Theta$. Then, for all $(\pi_1, \pi_2) \in I_1\times I_2$,    and   $\rho \in \left( \rho_L(\Theta),\rho_U(\Theta)\right) $,  we have that:
\begin{eqnarray} \label{condition} 
	n(\pi_1, \pi_2, \lambda,\rho)   \leq  U(\Theta) = n(\bar{p}_1^{(0)}, \bar{p}_2^{(0)}, \lambda, \rho_U(\Theta))
\end{eqnarray}
Furthermore, for given $\bar{p}_1^{(0)}, \bar{p}_2^{(0)}, \lambda$, the sample size   $n(\bar{p}_1^{(0)}, \bar{p}_2^{(0)}, \lambda, \rho)$ is an increasing function of the correlation $\rho$.

The sample size given by $n(\bar{p}_1^{(0)}, \bar{p}_2^{(0)}, \lambda, \rho)$ delimits the values that the sample size could have in terms of the correlation   accounting for plausible deviations in the anticipated event rates.
If there is no prior information on the correlation, we can use $ U(\Theta)$ as the needed sample size. If otherwise, we have some prior information on the correlation value, the rationale used in  \ref{SSmethod}    using correlation categories   can be as well applied here to the function $n(\bar{p}_1^{(0)}, \bar{p}_2^{(0)}, \lambda, \rho)$.  
Table \ref{Table.SampleSizeApproachBands} provides the sample size strategy under this circumstance. We lay out the performance of the sample size when varying the event rates in the intervals $I_1$ and $I_2$ and the subsequent sample size behavior according to the correlation  in Propositions 2 and 3 in the supplementary material.

{
	\begin{table}[h!] 
		\small
		\centering
		\caption[position=above]{
			Correlation category and its subsequent   correlation bounds,  $ \rho_L(\cdot) $ and $ \rho_U(\cdot) $ 
			for the intervals of plausible values for event rates, $I_1=\left[ \underline{p}_1^{(0)},\bar{p}_1^{(0)} \right]$ and $I_2=\left[ \underline{p}_2^{(0)},\bar{p}_2^{(0)} \right]$, and marginal effect sizes $\lambda=(\delta_1,\delta_2)$, and where $\Theta=(I_1,I_2,\lambda)$ denotes the set of values for the marginal components. Sample size bounds   for each correlation category and proposed sample size strategy calculated by \eqref{nCBE}   according to the intervals $I_1$ and $I_2$,   the marginal effect sizes $\lambda$, for given significance level $\alpha$ and power $1-\beta$. }
		\begin{tabular}{cccc}
			\toprule 
			Category 
			&  Correlation Bounds 
			& Sample Size Bounds
			&  Chosen Sample Size     \\
			\toprule 
			Weak & $  \left[  \rho_L(\Theta), \frac{\rho_U(\Theta)-\rho_L(\Theta)}{3} \right)  $ 
			& $\Big[ n(\underline{p}_1^{(0)}, \underline{p}_2^{(0)}, \lambda,  \rho_L(\Theta)) , \ \ n\left( \bar{p}_1^{(0)}, \bar{p}_2^{(0)}, \lambda,\frac{\rho_U(\Theta)-\rho_L(\Theta)}{3} \right)  \Big]$ 
			& $ n\left(  \bar{p}_1^{(0)}, \bar{p}_2^{(0)}, \lambda, \frac{\rho_U(\Theta)-\rho_L(\Theta)}{3}\right)$ \\ [4mm]     
			Moderate & $ \left[   \frac{\rho_U(\Theta)-\rho_L(\Theta)}{3},  \frac{2(\rho_U(\Theta)-\rho_L(\Theta))}{3}  \right)  $ 
			& $  \left[   n\left(\underline{p}_1^{(0)}, \underline{p}_2^{(0)}, \lambda, \frac{\rho_U(\Theta)-\rho_L(\Theta)}{3}\right) , \ \  n\left( \bar{p}_1^{(0)}, \bar{p}_2^{(0)}, \lambda, \frac{2(\rho_U(\Theta)-\rho_L(\Theta))}{3} \right) \right] $
			& $ n\left( \bar{p}_1^{(0)}, \bar{p}_2^{(0)} , \lambda,\frac{2(\rho_U(\Theta)-\rho_L(\Theta))}{3}\right) $  \\ [4mm]     
			Strong  
			& $  \left[   \frac{2(\rho_U(\Theta)-\rho_L(\Theta))}{3}, \rho_U(\Theta)  \right)  $ 
			& $ \Big[  n\left( \underline{p}_1^{(0)}, \underline{p}_2^{(0)}, \lambda, \frac{2(\rho_U(\Theta)-\rho_L(\Theta))}{3}\right) , \ \ n(\bar{p}_1^{(0)}, \bar{p}_2^{(0)}, \lambda,\rho_U(\Theta)) \Big]$
			&   $n(\bar{p}_1^{(0)}, \bar{p}_2^{(0)}, \lambda,\rho_U(\Theta))$  \\
			[4mm]   
			No prior information  
			& $  \left[  \rho_L(\Theta), \rho_U(\Theta) \right)  $ 
			& $\Big[ n\left( \underline{p}_1^{(0)}, \underline{p}_2^{(0)}  ,  \lambda,\rho_L(\Theta)\right), \ \ n\left( \bar{p}_1^{(0)}, \bar{p}_2^{(0)} , \lambda,\rho_U(\Theta)\right)  \Big]$ 
			&  $ n\left(\bar{p}_1^{(0)}, \bar{p}_2^{(0)} , \lambda, \rho_U(\Theta)\right)$ \\
			\bottomrule
		\end{tabular} 
		\label{Table.SampleSizeApproachBands}
	\end{table}
}

\subsection{Power performance of the proposed strategies }

Given   ($\theta, \lambda, \rho$)  and for a fixed sample size $N$, the power function using the unpooled variance estimate is defined as:
\begin{eqnarray} \label{power_unpooled}
	\psi(\theta, \lambda, \rho, N) &=& \Phi \left( \frac{ \sqrt{N} \cdot \delta_*(\theta, \lambda, \rho)   }{  \sqrt{p_*^{(0)}(\theta,\rho)\left(  1-p_*^{(0)}(\theta,\rho)\right)  +  \left(p_*^{(0)}(\theta,\rho) + \delta_*(\theta, \lambda, \rho) \right) \left(1-p_*^{(0)}(\theta,\rho) - \delta_*(\theta, \lambda, \rho)\right) }  } - z_\alpha \right) 
\end{eqnarray}
where $\Phi(\cdot)$ denotes the cumulative distribution of the standard normal distribution. The power function for the pooled variance estimator can be found  in the online support material.

In what follows, we show that the planned power $1-\beta$ is achieved with any of the previous  strategies in 
Subsections \ref{SSmethod} and \ref{Section.Bands}.
\begin{itemize}
	\item    If $\theta$ and $\lambda$ are fixed and the correlation value is  not  known,  we have   $n(\theta,\lambda,\rho)\leq  n\big(\theta,\lambda,B_U(\theta,\lambda)\big)$ and the proposed sample size becomes $N=n\big(\theta,\lambda,B_U(\theta,\lambda)\big)$. The  resulting power is then such that:
	$$\psi\left(\theta, \lambda, \rho, n\big(\theta, \lambda,B_U(\theta, \lambda)\big) \right)\leq\psi\left(\theta, \lambda, \rho, n(\theta, \lambda,\rho) \right).$$
	The power attained using the upper bound of the correlation is equal to the pre-specified power value ($1-\beta$) when the  correlation $\rho$ is the maximum value within its range, that is, $B_U(\theta, \lambda)$. Otherwise, if the correlation is less than $B_U(\theta, \lambda)$, the power will be always higher than the pre-specified power. Table S1 
	in the online supplementary material details  the power performance when the correlation  categories are taken into account.
	\item   If the event rate value $p_k^{(0)}$ is within the interval $I_k$ for $k=1,2$ and the effect sizes
	$\lambda$ are fixed, then $n(p_1^{(0)}, p_2^{(0)}, \lambda,\rho)  \leq  n\big(\bar{p}_1^{(0)}, \bar{p}_2^{(0)}, \lambda, \rho\big)$.	
	If in addition  we have no prior information on the correlation value, then   since the  sample size increases with respect to the correlation, it follows that $ n\big(\bar{p}_1^{(0)}, \bar{p}_2^{(0)}, \lambda, \rho\big) \leq n\big(\bar{p}_1^{(0)}, \bar{p}_2^{(0)}, \lambda, \rho_U(\Theta)\big)$, and then 	
	the proposed sample size turns into $N= n\big(\bar{p}_1^{(0)}, \bar{p}_2^{(0)}, \lambda, \rho_U(\Theta)\big)$. The corresponding power then satisfies:
	$$\psi\left(\theta, \lambda, \rho, n\big( \bar{p}_1^{(0)}, \bar{p}_2^{(0)}, \lambda, \rho_U(\Theta)\big)  \right)\leq\psi\left(\theta, \lambda, \rho, n(p_1^{(0)}, p_2^{(0)}, \lambda,\rho) \right).$$
	The power attained is equal to the pre-specified power value when the event rates $p_k^{(0)}$ take the upper values $\bar{p}_k^{(0)}$ and the correlation $\rho$ is equal to $\rho_U(\Theta)$. If that is not the case, the   power obtained  will be larger than the pre-specified $1-\beta$. 
\end{itemize}


\section{Motivating example: TACTICS-TIMI $18$ trial}
\label{Section.casestudy} 

In managing the syndrome of unstable angina and non-Q-wave acute myocardial infarction, there is controversy over whether using an invasive strategy rather than a conservative strategy offers any advantage. TACTICS-TIMI $18$ was a randomized trial that evaluated the efficacy of invasive and conservative treatment strategies in patients with unstable angina and non-Q-wave AMI treated with tirofiban, heparin, and aspirin \cite{Cannon2001}.  

Patients were randomly assigned to either an early invasive strategy or an early conservative strategy. The primary hypothesis of the TACTICS-TIMI $18$ trial was that an early invasive strategy would reduce the combined incidence of death, acute myocardial infarction, and rehospitalization for acute coronary syndromes at six months when compared with an early conservative strategy. The primary endpoint was the composite endpoint formed by a combination of incidence of death or  non-fatal myocardial infarction  ($\varepsilon_1$), and rehospitalization for acute coronary syndrome ($\varepsilon_2$) at six months.

For illustrative purposes, we assume that a trial will be planned for a similar setting and that the results of TACTICS-TIMI 18 are to be used. Since   previous   studies to  TACTICS-TIMI 18 also considered the events death and  non-fatal myocardial infarction altogether, we presume that the event rate and effect size on the endpoint $\varepsilon_1$ can be anticipated despite being composed by two events.
The estimated values for the frequency of death or  non-fatal myocardial infarction  ($\varepsilon_1$)  in the conservative strategy group was $\hat{p}_1^{(0)}=0.095 $  with a standard deviation of $0.009$; whereas the frequency of rehospitalization for acute coronary syndrome ($\varepsilon_2$) was  $\hat{p}_2^{(0)}= 0.137$  with a   standard deviation of $0.010$.
Based on the standard deviations of the  estimated event rates, we use the $95\%$ confidence intervals  as a set of     plausible values among which  the true values $p_1^{(0)}$, $p_2^{(0)}$ take values, that is, $I_1= [0.078, 0.112]$ and $I_2= [0.117, 0.157]$.
The observed effects on TACTICS-TIMI $18$  were  $ \delta_1=-0.022$ and  $ \delta_2 = -0.027$,  and we will use these as the expected effects on the new experimental trial.

We consider these parameters to construct the correlation bounds outlined in equation \eqref{corr.bounds}. The effects $\delta_1$ and $\delta_2$ and the   values $\hat{p}_1^{(0)} $ and $\hat{p}_2^{(0)}$ imply that  the eligible values for $ \rho $ lie in the interval ($-0.10$, $0.80$).
Using the intervals $I_1$ and $I_2$,  the  correlation bounds are such that the considered   values are plausible for any event rate within $I_1$ and $I_2$. This gives us the correlation bounds ($-0.08$, $0.77$).
Table \ref{TableFigure_CaseStudy} and its accompanying figure show  the correlation bound according to $\delta_1$ and $\delta_2$ with varying   values of the event rates. Observe that the upper bound takes the value $1$ when both event rates are equal, and the lower bound tends to $0$ when at least one of the event rates becomes smaller.

\begin{table}[h!]
	\centering
	\parbox{0.45\textwidth}{
		\begin{center}
			\begin{tabular}{cc}
				\hline 
				\textbf{Event rate values} &
				\textbf{Correlation Bounds}  \\
				\toprule   
				$\hat{p}_1^{(0)}=0.095$, \ $\hat{p}_2^{(0)}=0.137$ &  $ -0.10 \leq \rho \leq 0.80 $  \\[2mm]  
				$\bar{p}_1^{(0)}=0.112$, \ $\bar{p}_2^{(0)}=0.157$ &  $ -0.12 \leq \rho \leq 0.81 $    \\ [2mm]  
				$\underline{p}_1^{(0)}=0.078$, \ $\underline{p}_2^{(0)}=0.117$ &   $ -0.08 \leq \rho \leq 0.77 $     \\
				\bottomrule
			\end{tabular}
		\end{center}
		\caption{Lower bound, $B_L(\theta,\lambda)$,     and upper bound, $B_U(\theta,\lambda)$,   for the correlation according  to the   effect sizes $ \delta_1=-0.022$, $ \delta_2 = -0.027$  and for different values of the event rates.\newline
			\textbf{FIGURE  }	Lower bound  (surface  in blue) and upper bound  (in red)    for the correlation according to the   effect sizes $ \delta_1=-0.022$, $ \delta_2 = -0.027$ and where the marginal event rates take values between $0$ and $0.2$.}
		\label{TableFigure_CaseStudy}
	}
	\qquad
	\begin{minipage}[c]{0.5\textwidth}%
		\centering
		\includegraphics[width=0.9\textwidth]{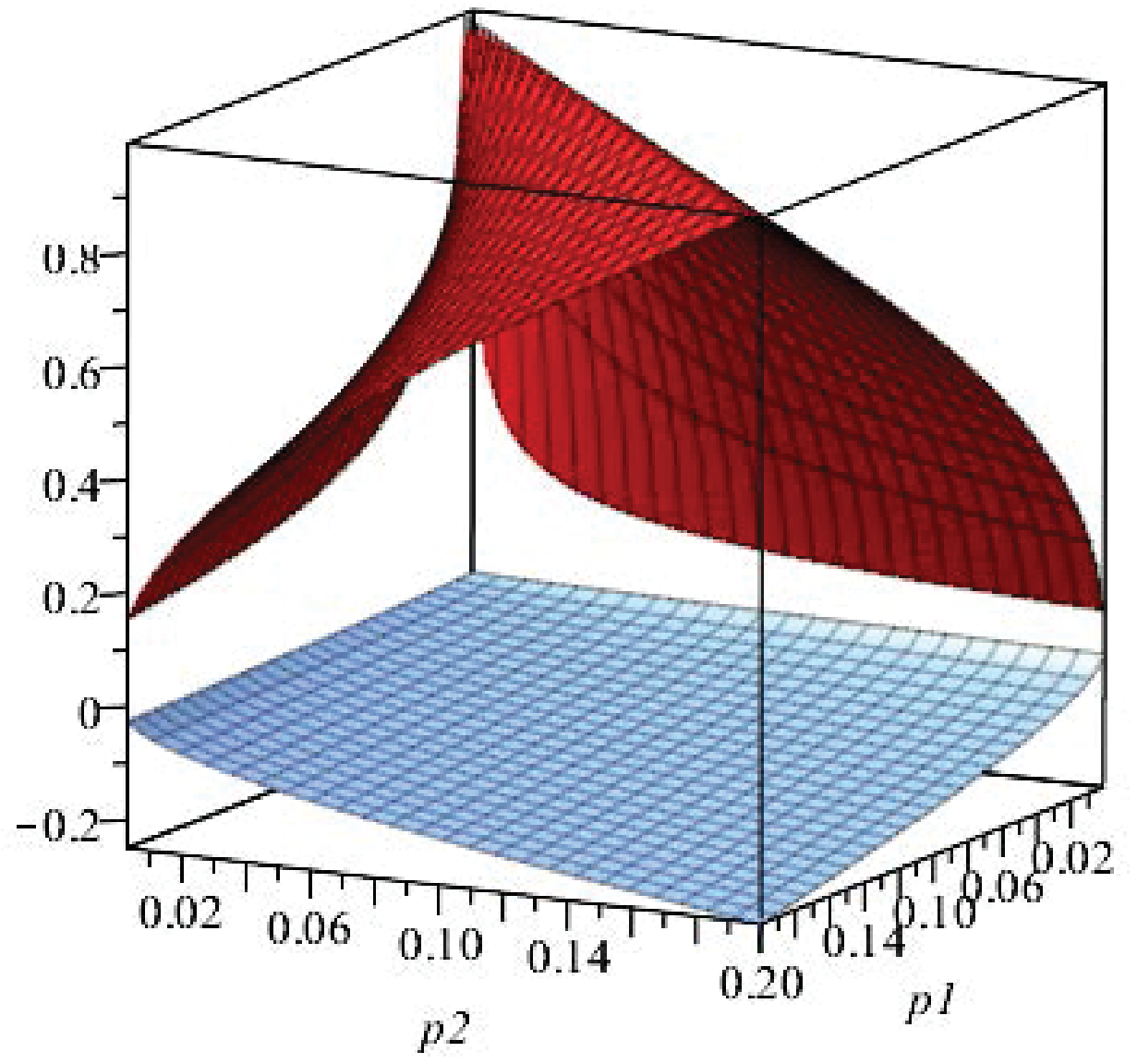}
	\end{minipage}
\end{table}

We illustrate the aspects of calculating power and sample size using the platform CompARE. CompARE calculates the sample size by anticipating the marginal information in terms of either risk difference, relative risk, or  odds ratio. In this particular case, we use the statistical test for risk difference under  pooled variance in order to ascertain the treatment differences in the composite endpoint at a significance level  of  $\alpha=0.025$ and target power of $1-\beta = 0.80$. The results obtained from CompARE are presented in the form of summary tables and plots.

Figure \ref{Plot_samplesizepower2} (left panel) depicts the performance of the sample size in terms of the correlation for given marginal parameters $\theta= (\hat{p}_1^{(0)}, \hat{p}_2^{(0)})$ and $\lambda=(\delta_1, \delta_2)$; and it illustrates the recommended sample size for each correlation category (weak, moderate, and strong).  
The solid line represents the sample size as a function of the correlation computed for the anticipated values $\theta$, and the shaded areas represent the region of values, constructed by   $I_1$, $I_2$, $ \delta_1$ and $ \delta_2$, within which interval the sample size falls. 
Based on  $I_1$ and $I_2$ the proposed sample size (in dotted lines)  is the upper value of the shaded area within the correlation category.  

Note that the sample size is highly sensitive to the anticipated parameters. For instance, for $\rho=0.3$,
using   $\hat{p}_1^{(0)}$ and $\hat{p}_2^{(0)}$,   the required sample size is $n=3030$. This sample size, however, can differ substantially from that calculated using other reasonable values, such as the upper or lower limits for the intervals $I_1$ and $I_2$, which would imply  $n=2511$ and  $n=3540$, respectively.

Figure \ref{Plot_samplesizepower2} (right panel) describes the statistical power achieved under the proposed method. Assuming that we have correctly anticipated the correlation category, observe that   in all cases the achieved power is larger than the planned power, $1-\beta$. Then,   the method guarantees the desired power. If we could correctly anticipate the values of the event rates, then the achieved power would lie between $ 0.80$ and $0.87$, in accordance with the plausible correlation values. If we base the sample size calculation on the intervals $I_1$ and $I_2$,
we will be overestimating the statistical power more than in the previous case, thus obtaining a power between $ 0.80$ and $0.95$.

Table \ref{Table_CaseStudy}  describes the proposed sample size for each correlation category and reports the possible values for the statistical power, assuming that we have correctly anticipated the correlation category.

\begin{figure}[h!]
	\centering
	\includegraphics[width=1\linewidth]{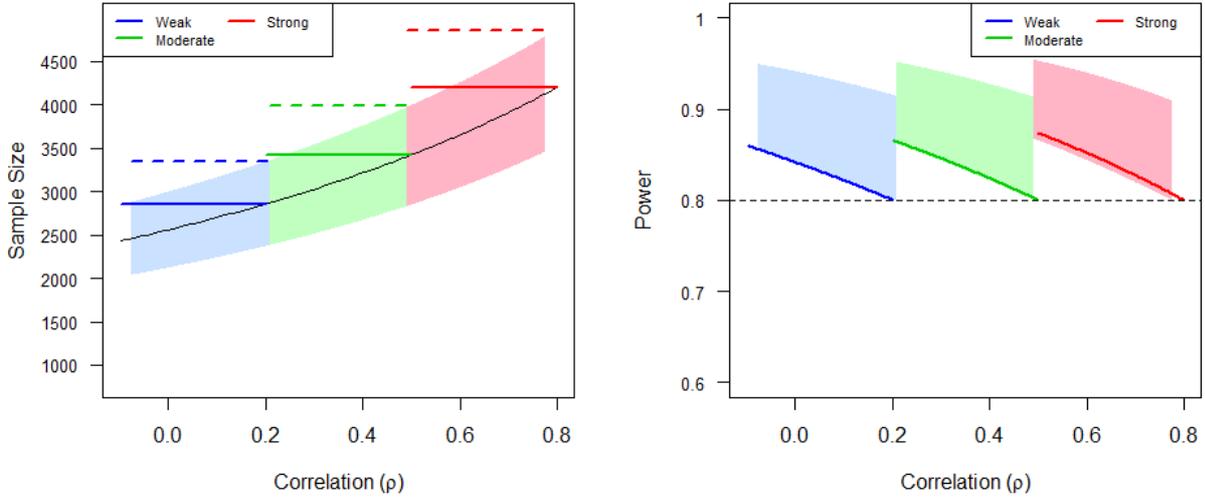}
	\caption[position=above]{Sample size (left panel) and power (right panel) as a function of the correlation according to the marginal effect sizes $ \delta_1=-0.022$ and $ \delta_2 = -0.027$;  either based on the  point values     $\hat{p}_1^{(0)}=0.095 $, $ \hat{p}_2^{(0)}= 0.137$ for the event rates  (solid line) or based on  the interval of plausible values for the event rates $I_1= [0.078, 0.112]$ and $I_2= [0.117, 0.157]$ (shaded areas).  The proposed sample size for each correlation category is highlighted in solid and dotted lines for, respectively, the point values and the interval values for the event rates. }
	\label{Plot_samplesizepower2}
\end{figure}

\begin{table}[h!]
	\centering
	\caption[position=above]{Recommended sample size for testing differences between the invasive strategy as compared with the conservative strategy. Underlying marginal parameters   are as follows:   $ p_1^{(0)}=0.095 $, $ p_2^{(0)}= 0.137$, $ \delta_1=-0.022$, $ \delta_2 = -0.027$.
		Both sample size and power were calculated based on the statistic \eqref{test.difprop.pooled} under the pooled variance  for a one-sided test at the significance level of $\alpha=0.025$. The given sample size was calculated to detect the effect on the composite endpoint with  the desired overall power of $1-\beta =0.80$. For calculating the power of the test, three sample size situations were considered, depending on the strength of the correlation: i)  weak correlation; ii) moderate correlation; iii) strong correlation. }
	\begin{tabular}{cccc}
		\hline 
		\multicolumn{4}{l}{Based on point values     $ p_1^{(0)}=0.095 $, $ p_2^{(0)}= 0.137$ for the event rates: }\\
		\multicolumn{4}{l}{Correlation bounds: $B_L(\theta, \lambda)=-0.10$, $B_U(\theta, \lambda)=0.80$.}\\
		\textbf{Association strength} &
		\textbf{Correlation} & \textbf{Sample size} & \textbf{Achieved power}  \\  
		\toprule   
		Weak &  $ -0.10 \leq \rho \leq 0.20 $  &   $2860$ & ($ 0.80 $, $0.86$)\\
		Moderate &  $ 0.20 < \rho \leq  0.50 $  &    $3425$ &  ($ 0.80 $, $0.87$) \\
		Strong &  $ 0.50 < \rho \leq  0.80 $  &   $4201$ &  ($ 0.80 $,$ 0.87 $) \\ 
		\hline
		\hline 
		\multicolumn{4}{l}{Based on intervals $I_1= [0.078, 0.112]$ and $I_2= [0.117, 0.157]$ for the event rates: }\\ 
		\multicolumn{4}{l}{Correlation bounds: $\rho_L(\Theta)=-0.08$, $\rho_U(\Theta)=0.77$.}\\
		\textbf{Association strength} &
		\textbf{Correlation} & \textbf{Sample size} & \textbf{Achieved power}  \\ 
		\toprule   
		Weak &  $ -0.08 \leq \rho \leq 0.21 $  &   $3355$ & ($ 0.80 $,   $0.95$)\\
		Moderate &  $ 0.21 < \rho \leq  0.49 $  &   $3970$ &  ($ 0.80 $,  $0.95$) \\
		Strong &  $ 0.49 < \rho \leq  0.77 $  &  $4782$ &  ($ 0.80 $, $0.95$) \\ 
		\bottomrule
	\end{tabular}
	\label{Table_CaseStudy}
\end{table}


\section{An extension for   risk ratio and  odds ratio}\label{Section.extension}

In this Section,  we show that the risk ratio and odds ratio for the composite endpoint can also be expressed in terms of its margins plus the correlation, and we extend the sample size derivation given in Section \ref{Section.ConvenientSampleSize} for evaluating  the   risk and odds ratio.

\subsection{Composite effect expressed in terms of the risk ratio or  the odds ratio}

Let $\mathrm{R}_k$ and $\mathrm{OR}_k$ denote the risk ratio and odds ratio, respectively,  for the $k$-th event. 
The risk ratio for the composite endpoint, $\mathrm{R}_*$, is expressed in terms of the risk ratio of its components $\mathrm{R}_1$ and $\mathrm{R}_2$, the event rates under control group, $p_1^{(0)}$ and $p_2^{(0)}$, and the correlation between them, $\rho$, as follows:
\begin{eqnarray} \label{effectCBE.riskratio}
	\mathrm{R}_* &=& \frac{p_1^{(0)}\mathrm{R}_1 + p_2^{(0)}\mathrm{R}_2- p_1^{(0)} p_2^{(0)}\mathrm{R}_1\mathrm{R}_2 - \rho  \sqrt{p_1^{(0)}\mathrm{R}_1  p_2^{(0)}\mathrm{R}_2 (1-p_1^{(0)}\mathrm{R}_1) (1-p_2^{(0)}\mathrm{R}_2) }}{1-  q_1^{(0)} q_2^{(0)} - \rho  \sqrt{p_1^{(0)} p_2^{(0)} q_1^{(0)} q_2^{(0)}}}
\end{eqnarray}
Analogously, the   odds ratio for the composite endpoint $\mathrm{OR}_*$ is  defined according to its margins and the correlation is given by: 
\begin{eqnarray} \label{effectCBE.oddsratio}
	\mathrm{OR}_* &=& 
	\frac{
		\left( \left(  1 +   \frac{\mathrm{OR}_1 p_1^{(0)}}{1-p_1^{(0)}}\right)  \left(  1 +  \frac{\mathrm{OR}_2 p_2^{(0)}}{1-p_2^{(0)}}\right) -1 - \rho  \sqrt{   \frac{\mathrm{OR}_1 \mathrm{OR}_2  p_1^{(0)} p_2^{(0)}}{(1-p_1^{(0)})(1-p_2^{(0)})}   }
		\right)\cdot \left( 1+ \rho  \sqrt{ \frac{p_1^{(0)} p_2^{(0)}}{(1-p_1^{(0)})(1-p_2^{(0)})}}\right) 
	}{ \left( \left( 1+ \frac{p_1^{(0)}}{(1-p_1^{(0)})}\right) \cdot \left( 1+ \frac{p_2^{(0)}}{(1-p_2^{(0)})}\right)  -1 - \rho  \sqrt{\frac{p_1^{(0)} p_2^{(0)}}{(1-p_1^{(0)})(1-p_2^{(0)})}}\right) \cdot \left( 1+ \rho  \sqrt{  \frac{\mathrm{OR}_1 \mathrm{OR}_2  p_1^{(0)} p_2^{(0)}}{(1-p_1^{(0)})(1-p_2^{(0)})}  }\right) }     
\end{eqnarray}	
The derivations of   equations  \eqref{effectCBE.riskratio} and \eqref{effectCBE.oddsratio}  are postponed to    Appendix \ref{app.treateffect}. 
By inspection of    \eqref{effectCBE.riskdiff}, \eqref{effectCBE.riskratio}, and \eqref{effectCBE.oddsratio}, we observe that 
if there is no effect on the components, that is, $\delta_1=\delta_2=1$, $\mathrm{R}_1=\mathrm{R}_2=1$ or $\mathrm{OR}_1=\mathrm{OR}_2=1$, then there is no effect on the composite endpoint, $\delta_*=\mathrm{R}_*=\mathrm{OR}_*=1$. However, the reciprocal does not follow: no effect on the composite endpoint is compatible with some effect on the components. Therefore, it is important to remark, as other authors have warned before \cite{Ferreira-Gonzalez2007,Ferreira-Gonzalez2007b,Tomlinson2010}, that not finding a beneficial effect on composite endpoint is not a guarantee of not having some effect on the components, hence the effect on the composite endpoint cannot be treated as if it were an indicator of some specific effect on its components.

\subsection{Sample size calculations in terms of risk ratio and odds ratio  }

The null hypothesis   in terms of the risk ratio is stated as  $\textrm{H}_0^*: \log(\mathrm{R}_*) = 0 $ and
the alternative hypothesis assuming  a risk reduction   is $\textrm{H}_1^*:\log(\mathrm{R}_*)<0$. 
The statistic that we use for testing the  
significance of the relative risk R$_*$ is:
\begin{eqnarray*} 
	\textrm{Z}_{*,n} &=&  \log(\widehat{\mathrm{R}}_* ) \big/ \sqrt{\widehat{Var}(\log(\widehat{\mathrm{R}}_* ))} 
\end{eqnarray*}
where   $ \widehat{\mathrm{R}}_* = \widehat{p}_*^{(1)}/\widehat{p}_*^{(0)}$. Under $\textrm{H}_0^*$, $ \textrm{Z}_{*,n} $ asymptotically follows the standard normal distribution; thus, we will reject $\textrm{H}_0^*$ at the $\alpha$ significance level if $ \textrm{Z}_{*,n} <-z_\alpha$.
As in Section \ref{Section.SampleSize},  we estimate  the variance $Var(\widehat{\mathrm{R}}_*)$ using the  pooled variance by means of 
$\widehat{Var}_{H_0}(\log(\widehat{\mathrm{R}}_* ))  = 
\frac{2}{n^{(0)} } \cdot \frac{ \widehat{q}_*^{(0)}+\widehat{q}_*^{(1)}}{ \widehat{p}_*^{(0)}+\widehat{p}_*^{(1)}} $  or by using the  unpooled variance, 
$\widehat{Var}_{H_1}(\log(\widehat{\mathrm{R}}_* ))  =    \frac{1}{n^{(0)}} \left(   \frac{1- \widehat{\mathrm{R}}_* \widehat{p}_*^{(0)} }{\widehat{\mathrm{R}}_* \widehat{p}_*^{(0)}} + \frac{\widehat{q}_*^{(0)} }{\widehat{p}_*^{(0)}}   \right) $.

For a given probability under control group $p_*^{(0)}$, and a significance level $\alpha$,  the needed sample size for detecting a risk ratio $\Gamma_*=p^{(1)}_*/p^{(0)}_*$ with power $1-\beta$ is given by: 
\begin{eqnarray} \label{samplesize.R.unpooled} 
	n &=&  2\cdot   \left( z_\alpha   + z_\beta  \right)^2 \cdot  \left(  \frac{1- \Gamma_* p_*^{(0)} }{\Gamma_* p_*^{(0)}} + \frac{q_*^{(0)} }{p_*^{(0)}}  \right) \Bigg/ \log(\Gamma_*)^2 
\end{eqnarray} 
The corresponding sample size when the pooled variance is used can be seen in  Table \ref{Table.SampleSizes}.

When measuring the effect of treatment with the odds ratio, the  null hypothesis     $\textrm{H}_0^*: \log(\mathrm{OR}_*) = 0 $ is compared with
the alternative hypothesis   $\textrm{H}_1^*:\log(\mathrm{OR}_*)<0$. To test the above hypotheses we use the statistic:
\begin{eqnarray*}
	\textrm{W}_{*,n} &=& \log(\widehat{\mathrm{OR}_*}) \big/  \sqrt{\widehat{Var}(\log(\widehat{\mathrm{OR}_*}))  } 
\end{eqnarray*}
where  $ \widehat{\mathrm{OR}}_* = \frac{\widehat{p}_*^{(1)}/\widehat{q}_*^{(1)}}{\widehat{p}_*^{(0)}/\widehat{q}_*^{(0)}}$ and 
where the pooled and unpooled variance estimates are given, respectively, by
$\widehat{Var}_{H_0}(\log(\widehat{\mathrm{OR}_*}))  =  \frac{8}{n^{(0)}  ( \widehat{p}_*^{(0)}+\widehat{p}_*^{(1)}) (\widehat{q}_*^{(0)}+\widehat{q}_*^{(1)}) }$  and 
$\widehat{Var}_{H_1}(\log(\widehat{\mathrm{OR}_*}))  =\frac{1}{n^{(0)}} \left( \frac{1}{  \hat{p}_*^{(0)} \hat{q}_*^{(0)}} + \frac{1}{  \hat{p}_*^{(1)} \hat{q}_*^{(1)}}\right) $.

Then the needed sample size is calculated using the unpooled variance, for detecting a treatment difference of $ \mathrm{OR}_*  = \Delta_* $ in order to have power $1-\beta$ at level $\alpha$ for given $p_*^{(0)}$, and it is given by: 
\begin{eqnarray} \label{samplesize.OR.unpooled}  
	n  &=&   2\cdot  \left( \frac{z_\alpha + z_\beta }{\log(\Delta_*)}\right)^2 \cdot \left(\frac{(  q_*^{(0)} +  p_*^{(0)}    \Delta_* )^2}{p_*^{(0)} q_*^{(0)}\Delta_*  } + \frac{1}{ p_*^{(0)} q_*^{(0)}} \right)  
\end{eqnarray}  
The sample size expression when using the pooled  variance   can be found in Table \ref{Table.SampleSizes}.

\begin{table}[h!] 
	\centering
	\caption[position=above]{ 
		Formulae for sample size determination   when comparing two treatments with respect to difference proportions, relative risks or odds ratio contrasts in a balanced design; where $n$ and $n^{(i)}$ denote  the total sample size and sample size per group ($i=0,1$) needed for testing the effect $\delta_*$, $\Gamma_*$ or $\Delta_*$ for a given event rate within control group $p_*^{(0)}$ at significance level $\alpha$ with $1-\beta$ power. }
	\begin{tabular}{ccc} 
		\toprule 
		& \textbf{Variance estimator} & \textbf{Sample Size formula\tnote{$ \star $}} \\ 
		\bottomrule 
		\textbf{Risk difference} & & \\
		Pooled variance &  $\frac{\left(  \widehat{p}_*^{(0)} + \widehat{p}_*^{(1)} \right)\left(  \widehat{q}_*^{(0)} + \widehat{q}_*^{(1)} \right) }{2n^{(0)}} $  &  
		$n =  2\cdot\left(  z_\alpha \cdot \sqrt{2 \bar{p}_*\bar{q}_*} + z_\beta \cdot \sqrt{p_*^{(0)}q_*^{(0)} + (p_*^{(0)}+\delta_*) (q_*^{(0)}-\delta_*) }   \right)^2 \Bigg/ \delta_*^{2}$  \\ 
		Unpooled variance &  
		$ \frac{\left( \widehat{p}_*^{(0)} \widehat{q}_*^{(0)} + \widehat{p}_*^{(1)} \widehat{q}_*^{(1)} \right)}{n^{(0)}  }   $
		& 
		$ n  =  2\cdot \left( z_\alpha + z_\beta \right)^2 \cdot \left( p_*^{(0)} q_*^{(0)} + (p_*^{(0)}+\delta_*) (q_*^{(0)}-\delta_*)\right) \Bigg/ \delta_*^{2}$    
		\\
		\hline
		\textbf{Risk ratio} &&\\
		Pooled variance &
		$\frac{2}{n^{(0)} } \cdot \frac{ \widehat{q}_*^{(0)}+\widehat{q}_*^{(1)}}{ \widehat{p}_*^{(0)}+\widehat{p}_*^{(1)}} $
		&
		$n =   2\cdot \left( z_\alpha \sqrt{ \frac{2 \bar{q}_*}{\bar{p}_*} } + z_\beta \sqrt{ \frac{1- \Gamma_* p_*^{(0)} }{\Gamma_* p_*^{(0)}} + \frac{q_*^{(0)} }{p_*^{(0)}}  } \right)^2 \Bigg/ \log(\Gamma_*)^2  $\\
		Unpooled variance &
		$\frac{1}{n^{(0)}} \left(   \frac{1- \widehat{\mathrm{R}}_* \widehat{p}_*^{(0)} }{\widehat{\mathrm{R}}_* \widehat{p}_*^{(0)}} + \frac{\widehat{q}_*^{(0)} }{\widehat{p}_*^{(0)}}   \right)$
		&
		$n \  = \  2\cdot \left( z_\alpha   + z_\beta  \right)^2 \cdot  \left(  \frac{1- \Gamma_* p_*^{(0)} }{\Gamma_* p_*^{(0)}} + \frac{q_*^{(0)} }{p_*^{(0)}}  \right) \Bigg/ \log(\Gamma_*)^2 $\\
		\hline 
		\textbf{Odds ratio } & &\\
		Pooled variance &
		$\frac{8}{n^{(0)}  ( \widehat{p}_*^{(0)}+\widehat{p}_*^{(1)}) (\widehat{q}_*^{(0)}+\widehat{q}_*^{(1)}) } $
		&
		$n =  2\cdot \left(  z_\alpha \sqrt{\frac{2}{ \bar{p}_*^{(1)} \bar{q}_*^{(1)}} } +   z_\beta \cdot \sqrt{ \frac{(  q_*^{(0)} +  p_*^{(0)}    \Delta_* )^2}{p_*^{(0)} q_*^{(0)}\Delta_*  } + \frac{1}{ p_*^{(0)} q_*^{(0)}}} \right)^2  \Bigg/ \log(\Delta_*)^2$
		\\
		Unpooled variance & 
		$ \frac{1}{n^{(0)}} \left( \frac{1}{  \hat{p}_*^{(0)} \hat{q}_*^{(0)}} + \frac{1}{  \hat{p}_*^{(1)} \hat{q}_*^{(1)}}\right) $
		&
		$ n  =  2\cdot \left(z_\alpha + z_\beta\right)^2 \cdot \left(\frac{(  q_*^{(0)} +  p_*^{(0)}    \Delta_* )^2}{p_*^{(0)} q_*^{(0)}\Delta_*  } + \frac{1}{ p_*^{(0)} q_*^{(0)}} \right) \Bigg/ \log(\Delta_*)^2$ \\
		\bottomrule 
	\end{tabular} 
	\begin{tablenotes}
		\item[$ \star $] \ where: $\bar{p}_*=\frac{p_*^{(0)}+p_*^{(1)}}{2}$ and $\bar{q}_*=\frac{q_*^{(0)}+q_*^{(1)}}{2}$.
	\end{tablenotes}
	\label{Table.SampleSizes}
\end{table}

\subsection{Sample size derivation based on its margins} \label{Section.DerivationExt}

Analogously to Section \ref{Section.ConvenientSampleSize} and following the notation in Section 4.4,, we  obtain the sample size based on  the   risk ratio as a function of the marginal effects $\mathrm{R}_1$ and $\mathrm{R}_2$,   the event rates $\theta$, and the correlation $\rho$. To do so, we take the event rate and   risk ratio of the composite endpoint for their expressions (which are defined according to $\theta$, $\mathrm{R}_1$, $\mathrm{R}_2$  and $\rho$, see equations \eqref{probCBE} and \eqref{effectCBE.riskratio}), and then substitute these into the sample size formula in \eqref{samplesize.R.unpooled}.
We denote by  $n(\theta,\mathrm{R}_1, \mathrm{R}_2, \rho)$  the needed sample size for evaluating the risk ratio computed for specific values   $\theta$, $\mathrm{R}_1, \mathrm{R}_2$, and  $\rho$.
We will analogously proceed with sample size in terms of  the odds ratio using the   effects $\mathrm{OR}_1$ and $\mathrm{OR}_2$, then denote by $n(\theta,\mathrm{OR}_1, \mathrm{OR}_2, \rho)$  the corresponding sample size.

In what follows, we describe the performance of the sample size  when the effect is measured by odds ratio or risk ratio. Further details of these properties and their empirical proof are to be found  in the web supplementary material.
\begin{itemize} 
	\item For fixed ($\theta,\mathrm{R}_1, \mathrm{R}_2$) or 
	($\theta, \mathrm{OR}_1, \mathrm{OR}_2$), the   sample size   for testing the effect measured by the risk ratio, $n(\theta, \mathrm{R}_1, \mathrm{R}_2, \rho)$, and the   sample size  for testing the odds ratio,   $n(\theta, \mathrm{OR}_1, \mathrm{OR}_2, \rho)$,  are   increasing functions of the correlation $\rho$. 
	\item 
	For given $\mathrm{R}_1$ and $\mathrm{R}_2$   at fixed $\rho=r$, the needed sample size $n(p_1^{(0)},p_2^{(0)},\mathrm{R}_1, \mathrm{R}_2,r)$ to have power $1-\beta$ at a significance level $\alpha$ falls into the interval:
	\begin{eqnarray} \label{CI_ss.Ext}  
		\mathbf{I}(r, I_1, I_2, \mathrm{R}_1, \mathrm{R}_2 )=
		[ \ n(\bar{p}_1^{(0)}, \bar{p}_2^{(0)}, \mathrm{R}_1, \mathrm{R}_2, r), \ \  
		n(\underline{p}_1^{(0)}, \underline{p}_2^{(0)},  \mathrm{R}_1, \mathrm{R}_2, r)  \ ]
	\end{eqnarray}
	Analogously, for given $\mathrm{OR}_1$ and $\mathrm{OR}_2$,  the needed sample size $n(p_1^{(0)},p_2^{(0)},\mathrm{OR}_1, \mathrm{OR}_2,r)$ lies within the interval:
	\begin{eqnarray*}  
		\mathbf{I}(r, I_1, I_2, \mathrm{OR}_1, \mathrm{OR}_2 )=
		[ \ n(\bar{p}_1^{(0)}, \bar{p}_2^{(0)}, \mathrm{OR}_1, \mathrm{OR}_2, r), \ \  
		n(\underline{p}_1^{(0)}, \underline{p}_2^{(0)},  \mathrm{OR}_1, \mathrm{OR}_2, r)  \ ]
	\end{eqnarray*}
	\item 	
	For all $(\pi_1, \pi_2) \in I_1\times I_2$    and   $\rho \in \left( \rho_L(\Theta),\rho_U(\Theta)\right) $, it follows that:
	\begin{eqnarray*} 
		n(\pi_1, \pi_2, \mathrm{R}_1, \mathrm{R}_2,\rho) & \leq & \mathcal{U}_{R}(\Theta) = n(\underline{p}_1^{(0)}, \underline{p}_2^{(0)},  \mathrm{R}_1, \mathrm{R}_2, \rho_U(\Theta)) \\
		n(\pi_1, \pi_2, \mathrm{OR}_1, \mathrm{OR}_2,\rho) & \leq & \mathcal{U}_{OR}(\Theta) = n(\underline{p}_1^{(0)}, \underline{p}_2^{(0)},  \mathrm{OR}_1, \mathrm{OR}_2, \rho_U(\Theta)) 
	\end{eqnarray*}
	Furthermore, for given ($\underline{p}_1^{(0)}, \underline{p}_2^{(0)},\mathrm{R}_1, \mathrm{R}_2$) or  ($\underline{p}_1^{(0)}, \underline{p}_2^{(0)}, \mathrm{OR}_1, \mathrm{OR}_2$), the sample size functions  $n(\underline{p}_1^{(0)}, \underline{p}_2^{(0)},  \mathrm{R}_1, \mathrm{R}_2, \rho ) $ and $n(\underline{p}_1^{(0)}, \underline{p}_2^{(0)},  \mathrm{OR}_1, \mathrm{OR}_2, \rho ) $ increase with respect to the   correlation $\rho$.  
\end{itemize}

Note that,  unlike  when using risk differences, the sample size has its maximum value when both event rates take their lower interval values $\underline{p}_1^{(0)}, \underline{p}_2^{(0)}$ (see equations \eqref{CI_ss} and \eqref{CI_ss.Ext}).

Also note that if the marginal parameters ($\theta, \mathrm{R}_1, \mathrm{R}_2$) or ($\theta, \mathrm{OR}_1, \mathrm{OR}_2$) are anticipated and the correlation is not known,  the sample size strategy described in Section \ref{SSmethod}  can be extended to the risk   and odds ratio and analogously  applied.
For fixed effects ($\mathrm{R}_1,\mathrm{R}_2$) or ($\mathrm{OR}_1,\mathrm{OR}_2$), and given intervals $I_1$ and $I_2$ for the event rates, we can follow  the same reasoning as for risk differences  in Section \ref{Section.Bands}, and use $\mathcal{U}_R(\Theta)$ (analogously $\mathcal{U}_{OR}(\Theta)$) to calculate  the required sample size that  guarantees the planned power while accounting for the unknown correlation value and uncertainty of the marginal parameter values.

\section{A simulation study} \label{Section.simulations}

We conduct a simulation study to evaluate  the strategies proposed in Section \ref{Section.ConvenientSampleSize} for computing the sample size.

\subsection{Design}
We simulate  a  two-arm trial with a composite primary endpoint composed of two events,  $\varepsilon_1$ and  $\varepsilon_2$, according to the following values (which are all summarized in Table \ref{Table_Scenarios}): the marginal probabilities of observing 
$\varepsilon_k$ ($k=1,2)$ in the  control  group $\theta=(p_1^{(0)}, p_2^{(0)})$ take values between $0.01$ and $ 0.2$, and they  are without loss of generality such that $p_1^{(0)}< p_2^{(0)}$; the risk ratios $\lambda =( \mathrm{R}_1, \mathrm{R}_2 )$ are specified  for beneficial effects  and vary from $0.6$ to $0.8$;   the true correlation between $\varepsilon_1$ and  $\varepsilon_2$ is assumed to be common for both groups, and it covers the  positive range between $0$ and $1$.  The possible combinations add up to a total of 421 different scenarios which take into account that for  given  $(\theta,\lambda)$. simulations are   performed only for those $\rho_{true}$ between $B_L(\theta,\lambda)$ and $B_U(\theta,\lambda)$ (see \eqref{corr.bounds}).

For each one of these 421 scenarios  specified by ($\theta$, $\lambda$, $\rho_{true}$),  we compute the required sample size $n(\theta, \lambda, \rho(\theta,\lambda))$ for a one-sided test with   power   $1 -\beta =0.80$ at the significance level $\alpha=0.025$, which is done by following one of the six  different formulations that are  derived in Section \ref{Section.SampleSize}  and  Section  \ref{Section.extension} and, additionally, are all summarized in Table \ref{Table.SampleSizes}.

We distinguish 4 different situations according to the value  we  assume in  $\rho(\theta,\lambda)$ to calculate $n(\theta, \lambda, \rho(\theta,\lambda))$:  
\begin{enumerate}
	\item
	For the weak correlation category, use $ \rho(\theta,\lambda) = B_U(\theta,\lambda)/3 $
	\item For the moderate correlation  category, use $ \rho(\theta,\lambda) = 2 B_U(\theta,\lambda)/3 $
	\item For the strong correlation category, use $ \rho(\theta,\lambda) = B_U(\theta,\lambda) $
	\item For guessing the true correlation, use $ \rho(\theta,\lambda) = \rho_{true}$.
\end{enumerate}

Given one scenario  specified by ($\theta$, $\lambda=(R_1, R_2)$, $\rho_{true}$), we evaluate the type I error  first by calculating $n$ based on   $(\theta, \lambda=(R_1, R_2), \rho(\theta,\lambda))$ and simulating  $100000$ trials using  $(\theta, \lambda=(1, 1), \rho_{true}$).
To check the power, we start by  calculating $n$  as above,  based on   $(\theta, \lambda=(R_1, R_2), \rho(\theta,\lambda))$,  and then we simulate  $100000$ trials using  $(\theta, \lambda=(R_1, R_2), \rho_{true}$). Altogether,  we have to  analyze  a total of 3368  scenarios.

The above steps have to be reproduced six times according to the different sample size formulae used to compute $n(\theta, \lambda, \rho(\theta,\lambda))$, that is, by stating the effect in terms of  the difference in proportions, the  risk ratios or  the  odds ratio, and using both the 
pooled and the  unpooled estimates of the variance.
We have performed all computations using the \texttt{R} software tool (Version 0.98.1087), and  the time required to perform the considered simulations was 55.58h.

\begin{table}[h!]
	\centering
	\caption[position=above]{{\bf Simulation scenarios}: Values of marginal event rates in the control group: $\theta= (p_1^{(0)} ,  p_2^{(0)})$;  treatment effects in terms of the risk ratio: $\lambda=( \mathrm{R}_1, \mathrm{R}_2)$; and correlation $ \rho_{true} $ between components. Note that not all the combinations are feasible because the correlation is between $B_L(\theta, \lambda)$ and $B_U(\theta,\lambda)$. }
	\begin{tabular}{cc} 
		\toprule   
		\textbf{Parameter} & \textbf{Values} \\ \toprule
		$ p_1^{(0)} $ & $0.01, 0.05, 0.10$  \\
		$ p_2^{(0)} $ & $0.01, 0.05, 0.10, 0.15, 0.20$ \\
		$ \rho_{true} $ & $0.0, 0.1, 0.2, 0.3, 0.4, 0.5, 0.6, 0.7, 0.8, 0.9, 1$  \\ \hline
		Effects used to evaluate the power: & \\
		$ \mathrm{R}_1, \mathrm{R}_2  $   & $0.6, 0.7, 0.8$ \\
		Effect used to evaluate the type I error: & \\ 
		$ \mathrm{R}_1= \mathrm{R}_2  $ & $1$ \\ 
		\bottomrule 
	\end{tabular}  
	\label{Table_Scenarios}
\end{table}

\subsection{Power analysis of the proposed strategies for computing sample size}

Let $n_{l,m}$ be the required sample size calculated using the formulae described in Table \ref{Table.SampleSizes} where $l=p,u$, indicating    whether the pooled or unpooled variance has been used; and  $m=D,R,OR$, indicating  the effect measure that has been tested. In other words, for the difference in proportions,  $m=D$; for relative risk, $m=R$; and for odds ratio, $m=OR$. 
Let   $\Psi_{l,m} $  denote   the empirical power    when the total number of participants is $n_{l,m}$ ($l=p,u$;  \ $m=D,R,OR$).

Whenever the correlation we are using to compute the sample size coincides with the one we have used to run the simulations ($ \rho(\theta,\lambda)=\rho_{true}$), the empirical powers are always  achieved whether we are using the pooled,  $\Psi_{p,m}$, or unpooled, $\Psi_{u,m}$, estimator of the variance. Nevertheless,   when 
testing the difference in proportions, the achieved powers do not 
substantially differ ($\Psi_{u,D}\cong\Psi_{p,D}$);
when testing the treatment differences in terms of the  risk ratio or the  odds ratio, the power achieved if  the unpooled variance estimator is  used is slightly larger than the power achieved with the pooled estimator,   $\Psi_{u,m}\leq\Psi_{p,m}$, $m=R, OR$ (see Table S3 
in supplementary material for a comparison of the  two approaches).  
The results presented herein refer to the unpooled variance estimator. The corresponding results for   the pooled variance are summarized in the supplementary material (Table S4 and Figure S1).

When $ \rho(\theta,\lambda) \neq \rho_{true} $, we distinguish two types of
misspecification.  Misspecification type I, $\rho_{true}$ and $\rho(\theta,\lambda)$   pertain to the same correlation category;
and Misspecification type II, $\rho_{true}$ and $\rho(\theta,\lambda)$  do not belong to the same category. 
Table \ref{Table2_EmpiricalPowerUNPOOLED} describes   the empirical power in  these two cases, which account for  the correlation category  for   the three    effect measures that we could use to test the difference between groups. 
If Misspecification I occurs, the pre-specified power is achieved and   might  exceed   $7\%$. 

For misspecification II,  there are two possible scenarios. The first is for those cases where the correlation $\rho(\theta, \lambda)$ is assumed in a stronger correlation category than   the one that $\rho_{true}$ belongs to, for instance, if $\rho(\theta, \lambda)$ is assumed to be strong and $\rho_{true}$ is moderate. Under this scenario, $\rho(\theta, \lambda)>\rho_{true}$, and then the planned power is always achieved. The second scenario is when the $\rho(\theta, \lambda)$ is assumed to be in a weaker correlation category than the one that $\rho_{true}$ lies in. For instance, when $\rho(\theta, \lambda)$ is assumed weak and $\rho_{true}$ is moderate. In those cases where $\rho(\theta, \lambda)<\rho_{true}$,   the trial will be underpowered.

The empirical power in terms of the difference between the assumed and true correlations is illustrated 
in Figure \ref{fig:scatterplotspower_Unpooled}. Observe  
that when the assumed correlation is greater than the true correlation, that is,  $\rho(\theta,\lambda)>\rho_{true}$,   the empirical power is equal to or greater  than the pre-specified power.
Note that in all cases under the  strong correlation category we have $\rho_{true} \leq \rho(\theta,\lambda)$, the pre-specified power is assured  even though  	we  failed to anticipate  	the    category.
Also note that there are no differences in the achieved power, nor are there any in the   method's performance in terms of the  measure we are  using to evaluate the effect. 

\begin{table}[h!]
	\centering
	\caption[position=above]{Median   empirical power, given the sample size (under the unpooled variance), depending on the misspecification error	and the assumed correlation.    
		Values in parentheses indicate the maximum and minimum of the empirical power.}
	\begin{tabular}{c|cc}
		\hline 
		\multirow{2}{*}{\textbf{Assumption}} &
		\textbf{Misspecification I:  } & \textbf{Misspecification II:  }  \\
		&
		\textbf{Correlation within the category} & \textbf{Correlation outside  the category}  \\ 
		\toprule  
		\textbf{Risk Difference} & & \\  
		Weak &   $ 0.82 $ $(0.80,0.86)$ &  $ 0.78 $ $(0.67,0.80)$ \\ 
		Moderate & $ 0.82 $ $(0.80,0.87)$ &  $ 0.82 $ $(0.74,0.91)$ \\ 
		Strong  & $0.82$ $(0.80,0.87)$  & $0.87$ $(0.81,0.95)$ \\
		\bottomrule
		\textbf{Risk Ratio} & & \\ 
		Weak &   $ 0.82 $ $(0.80,0.86)$ &  $ 0.78 $ $(0.67,0.81)$ \\ 
		Moderate & $ 0.82 $ $(0.80,0.87)$ &  $ 0.82 $ $(0.74,0.90)$ \\ 
		Strong  & $0.82$ $(0.80,0.87)$  & $0.88$ $(0.81,0.95)$ \\
		\bottomrule
		\textbf{Odds Ratio} & & \\   
		Weak &   $ 0.82 $ $(0.80,0.86)$ &  $ 0.78 $ $(0.67,0.81)$ \\ 
		Moderate & $ 0.82 $ $(0.80,0.87)$ &  $ 0.82 $ $(0.74,0.91)$ \\ 
		Strong  & $0.82$ $(0.80,0.87)$  & $0.87$ $(0.81,0.95)$ \\
		\bottomrule
	\end{tabular}  
	\label{Table2_EmpiricalPowerUNPOOLED}
\end{table}

\begin{figure}[h!]
	\centering
	\includegraphics[width=1.08\linewidth]{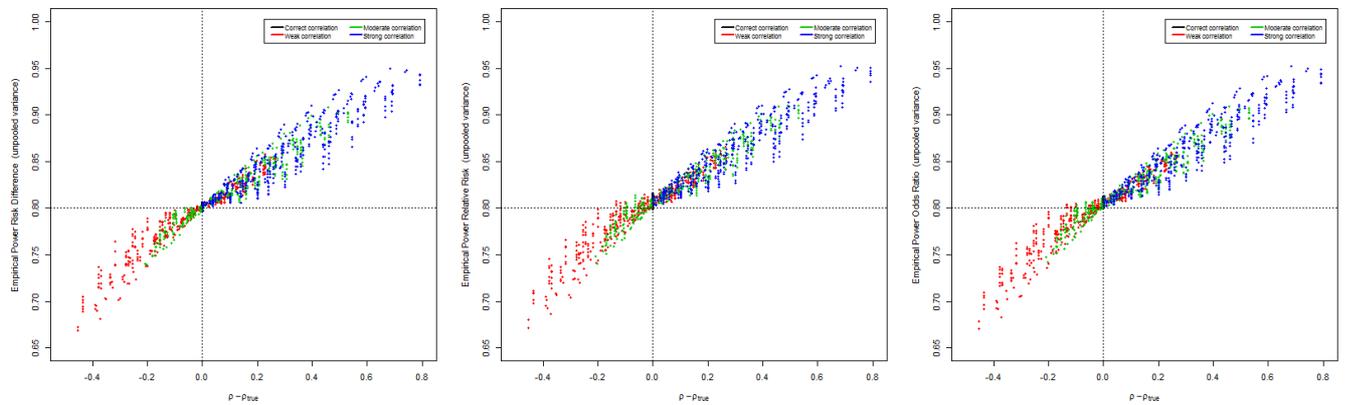}
	\caption{Scatterplot showing the relationship between   empirical power versus the difference between the assumed   and   true correlations for each of the sample size formulas (under unpooled variance) that were used in the simulation  study in section \ref{Section.simulations}.} 
	\label{fig:scatterplotspower_Unpooled} 
\end{figure}

\subsection{Type I error analysis of the proposed strategies for computing sample size }

The empirical results in the simulation study show that the   type I error is not affected by the misspecification of the correlation. Nonetheless, the   empirical type I error under  the pooled estimator may be slightly superior to    significance level $ 0.05$, especially when the treatment is tested in terms of risk ratio and odds ratio (see Figure S3 in the online support  material).


\section{Discussion}

Composite endpoints are increasingly used as primary endpoints to achieve greater incidence rates of observing the primary event, larger effect sizes and, hopefully, higher statistical power while avoiding multiplicity adjustment. Even so, their use creates   challenges in both the design and interpretation of the studies.

It is well known that sample size determination plays a key role in trial design. We have shown that calculating the sample size for composite binary endpoints needs more than the anticipated effect size and event rates of the composite components; it also needs the correlation between them. Sizing clinical trials in which composite endpoints are involved often implies facing the challenge of dealing with the unknown value of the correlation. We have assessed how much the correlation impacts the sample size  and, consequently, the attained power. Our conclusion is that the sample size strongly depends on the correlation and that the more correlation between the components are, the more sample size is needed. Motivated by this concern, we have proposed some strategies for deriving the sample size when the correlation is not specified. The strategy, based on the stratification of the correlation into different categories, assures the pre-specified power even without previous knowledge on the correlation. In addition,  if at least we could anticipate the  category where the correlation falls into, the achieved power would slightly surpass the planned power (see Table \ref{Table2_EmpiricalPowerUNPOOLED}). In those cases where not even  the correlation category can be anticipated, the interval of plausible values for the sample size might be too wide and  the proposed strategy might  be extremely conservative. Further research is needed in such cases to obtain more accurate power. 

We have illustrated our proposal using the platform CompARE. CompARE extends and includes a previous CompARE web page for composite time-to-event endpoints that was created for handling the relative efficiency of comparing treatment groups in terms of the composite endpoint versus one of its components as the primary endpoint. CompARE allows for the implementation of the Asymptotic Relative Efficiency method, which quantifies the gain in efficiency from using the composite endpoint over one of its components \cite{GomezLagakos, BofillGomez}, doing so specifically on the basis of information about the different endpoints and anticipated values.

Throughout this work, we have assumed that we are in the planning stage of a randomized clinical trial whose aim is to test the efficacy of a new treatment by comparing its performance with others that have already been approved. These trials are usually known to have much larger sample sizes. For this reason, we have restricted this work to sample size calculation based upon asymptotic approximations of the normal distribution. In previous trial phases devoted to   obtain  the optimal dose level or to  study the toxicity of the new drug, the sample size is not as large as in efficacy trials. In those cases, it could be more appropriate to base the sample size calculations on an exact test. Unlike the tests based on asymptotic distributions, the power function of an exact test  usually does not have an explicit form, and  the sample size  is obtained numerically by greedy search over the sample space. In practice, the applicability of such methods can come across difficulties because    intensive computing is required\cite{Chow2008}. There is controversy   over whether or not to use  exact tests, since when the sample size is not large enough,   the asymptotic test may not preserve the test size,   whereas exact tests could be conservative\cite{Fagerland2015,Crans2008,Andres2009}.

The sample size calculation in this work has been derived using the same correlations for both groups. Although this assumption is very often being used\cite{Sugimoto2017, Asakura2017, Ando2015},  it remains to be studied how plausible is in practice. We are working on an extension of our  methods   to account for different group correlations. Moreover, we are currently studying and implementing in CompARE other association measures for characterizing the strength of dependence between pairs of binary endpoints, such as the relative overlap \cite{Marsal2015},  which in practice might be easier to anticipate.

Interpreting the results of a trial with a primary composite endpoint is particularly challenging.  Composite endpoints comprise the information of its components and   capture a more complex picture of the intervention's efficacy, however, they might oversimplify  the evidence by looking only at the composite effect\cite{Pocock2015}. A proper study of the contribution from each separate component should be conducted to ensure a clear understanding of the results. What is more, composite endpoints are in many cases formed by a set of endpoints among whom  the clinical relevance highly differs. This could lead to misleading results about whether the treatment benefits only the less important endpoints. Moreover, as shown in Section \ref{Section.extension}, the effect for the composite does not necessarily reflect the effects for the components. CompARE computes the effect on the composite endpoint and gets constructive numerical and graphical results in order to investigate the role that each component plays.

Comparisons between two groups when using composite endpoints could be based either on the comparison of the corresponding proportions  by means, for instance, of a two-proportion z-test or based on the  comparison of  survival summaries  by means, for instance, of a log-rank test.   Although this paper focuses on composite binary endpoints, we want to emphasize that time-to-first-event  endpoints are  very often used  and that comparisons based on proportions test could be   less powerful than those based on full survival information \cite{Ryan1985,Buyse1982,Cuzick1982}.

Different strategies such as the win ratio\cite{Pocock2012, Luo2015} and the weighted combined approach\cite{Rauch2017} have been developed to    take into consideration the order of clinical priorities for the composite components when analyzing composite endpoints. Extending this work to more than two   components and  by incorporating weights remains open for future research.


\section*{Acknowledgments}
We would like to thank the referees for all their comments which have resulted in a clear improvement of the manuscript. We would also like to thank Dr A. Mart\'{i}n Andr\'{e}s   and the Spanish Network of Biostatistics \textit{Biostatnet}  (MTM2015-69068-REDT and MTM2017-90568-REDT).  

This work is partially supported by grants MTM2015-64465-C2-1-R (MINECO/FEDER) from the Ministerio de Econom\'{i}a y Competitividad (Spain), and 2017 SGR 622 (GRBIO) from the Departament d'Economia i Coneixement de la Generalitat de Catalunya (Spain). M. Bofill Roig acknowledges financial support from the Spanish Ministry of Economy and Competitiveness, through the Mar\'{i}a de Maeztu Programme for Units of Excellence in R\&D (MDM-2014-0445).


\subsection*{Financial disclosure}

None reported.

\subsection*{Conflict of interest}

The authors declare no potential conflict of interests.

\section*{Supporting information} 

\noindent
Additional figures and tables may be found online  in the supplementary document.
Source  code  for implementing and reproducing  the procedures discussed in this article is available at \url{https://github.com/MartaBofillRoig/CompARE}.

\vspace{10mm}

\appendix

\noindent
\textbf{Appendix}

Let $X_{ijk}$   denote the response  of the $k$-th binary
endpoint for the $j$-th patient in the $i$-th group of treatment ($
i=0,1 $, $j=1,...,n$, $k=1,2$). We denote by $X_{ij*}$ the composite response defined as
\begin{eqnarray} \label{CE_Def} 
	X_{ij*} &=&
	\begin{cases}
		1, \text{ if } X_{ij1} + X_{ij2} \geq 1 \\
		0, \text{ if else } X_{ij1} + X_{ij2} = 0
	\end{cases}
\end{eqnarray} 
W denote by $	 p_1^{(i)} = \mathrm{P}(X_{ij1} =1)=1- q_1^{(i)} $, 	$ p_2^{(i)} =
\mathrm{P}(X_{ij2} =1) =1- q_2^{(i)} $ and 	 $p_*^{(i)} = \mathrm{P}(X_{ij*} =1) =1- q_*^{(i)} $
the probabilities of observing  each endpoint  in the $i$-th  group.  
Let  $O_k^{(0)}, \delta_k, \mathrm{R}_k, \mathrm{OR}_k$ be the odds under the control group,
the risk difference, risk ratio and odds ratio, respectively, for the $k$-th endpoint, that is, 
$O_k^{(0)} = \frac{ p_k^{(0)}  }{ q_k^{(0)}   }$,
$\delta_k = p_k^{(1)}-p_k^{(0)}$,
$\mathrm{R}_k = \frac{p_k^{(1)}}{p_k^{(0)}}$, and
$\mathrm{OR}_k = \frac{p_k^{(1)}/q_k^{(1)} }{p_k^{(0)}/q_k^{(0)}}$.
We denote by $\theta= (p_1^{(0)},p_2^{(0)})$ the vector of marginal event rates, and $\lambda= (\delta_1,\delta_2)$ the vector of effect sizes.

Let $\rho^{(i)} $ represent the correlation between  $X_{ij1}$ and $X_{ij2}$ in the $i$-th treatment group, and $\rho$ refer to the correlation when it is  assumed to be equal in both groups, i.e., $\rho=\rho^{(0)}=\rho^{(1)}$.

\section{Derivation of the composite effect from the margins}\label{app.treateffect}

We derive the expression for the composite treatment effect  in terms of the marginal component and the correlation described in Sections  \ref{Section.Notation} and \ref{Section.extension},  and we prove the monotone performance of the risk difference with respect to the correlation $\rho$.

\begin{theorem}[Composite effect from margins]\label{thm_CBEeffect}
	The composite effect for the composite endpoint can be expressed in terms of the component parameters as follows:
	\begin{enumerate}[(i)]
		\item The risk difference for the composite endpoint, $\delta_*$, is
		determined by the  six parameters 
		$p_1^{(0)}$, $p_2^{(0)}$, $\delta_1$,
		$\delta_2$, $\rho^{(0)}$, $\rho^{(1)}$  and has the
		following expression:
		\begin{equation} \label{diff.prop}
			\delta_* = \delta_1 q_2^{(0)} + \delta_2 q_1^{(0)}  - \delta_1 \delta_2 + \rho^{(0)} \sqrt{p_1^{(0)}p_2^{(0)}q_1^{(0)}q_2^{(0)}} - \rho^{(1)} \sqrt{(p_1^{(0)}+\delta_1)(p_2^{(0)}+\delta_2)(q_1^{(0)}-\delta_1)(q_2^{(0)}-\delta_2)}
		\end{equation}
		
		\item The risk ratio for the composite endpoint, $\mathrm{R}_*$, is
		determined by the six parameters 
		$p_1^{(0)}$, $p_2^{(0)}$, $\mathrm{R}_1$,
		$\mathrm{R}_2$, $\rho^{(0)}$, $\rho^{(1)}$  and has the
		following expression: 
		\begin{equation}  \label{eqRR} 
			\mathrm{R}_* = \frac{p_1^{(0)}\mathrm{R}_1 + p_2^{(0)}\mathrm{R}_2- p_1^{(0)} p_2^{(0)}\mathrm{R}_1\mathrm{R}_2 - \rho^{(1)}  \sqrt{p_1^{(0)}\mathrm{R}_1  p_2^{(0)}\mathrm{R}_2 (1-p_1^{(0)}\mathrm{R}_1) (1-p_2^{(0)}\mathrm{R}_2) }}{1-  q_1^{(0)} q_2^{(0)} - \rho^{(0)}  \sqrt{p_1^{(0)} p_2^{(0)} q_1^{(0)} q_2^{(0)}}} 
		\end{equation}
		
		\item The odds ratio for the composite endpoint, $\mathrm{OR}_*$, is
		determined by the six parameters 
		$p_1^{(0)}$, $p_2^{(0)}$, $\mathrm{OR}_1$,
		$\mathrm{OR}_2$, $\rho^{(0)}$, $\rho^{(1)}$  and has the
		following expression: 
		\begin{equation} \label{eqOR} 
			\mathrm{OR}_* = 
			\frac{
				\left( \left(  1 +   \frac{\mathrm{OR}_1 p_1^{(0)}}{1-p_1^{(0)}}\right)  \left(  1 +  \frac{\mathrm{OR}_2 p_2^{(0)}}{1-p_2^{(0)}}\right) -1 - \rho^{(1)} \sqrt{   \frac{\mathrm{OR}_1 \mathrm{OR}_2  p_1^{(0)} p_2^{(0)}}{(1-p_1^{(0)})(1-p_2^{(0)})}   }
				\right)\cdot \left( 1+ \rho^{(0)} \sqrt{ \frac{p_1^{(0)} p_2^{(0)}}{(1-p_1^{(0)})(1-p_2^{(0)})}}\right) 
			}{ \left( \left( 1+ \frac{p_1^{(0)}}{(1-p_1^{(0)})}\right) \cdot \left( 1+ \frac{p_2^{(0)}}{(1-p_2^{(0)})}\right)  -1 - \rho^{(0)} \sqrt{\frac{p_1^{(0)} p_2^{(0)}}{(1-p_1^{(0)})(1-p_2^{(0)})}}\right) \cdot \left( 1+ \rho^{(1)} \sqrt{  \frac{\mathrm{OR}_1 \mathrm{OR}_2  p_1^{(0)} p_2^{(0)}}{(1-p_1^{(0)})(1-p_2^{(0)})}  }\right) }     
		\end{equation}	
		
	\end{enumerate} 
\end{theorem}

\begin{proof}[Proof of Theorem~\ref{thm_CBEeffect}]
	(i), (ii) The two expressions \eqref{diff.prop} and \eqref{eqRR} follow in a straightforward manner after noting that:
	\begin{eqnarray} \label{CE_prob} 
		p^{(i)}_{*} &=& 1-  q_1^{(i)} q_2^{(i)} - \rho^{(i)}    \sqrt{p_1^{(i)} p_2^{(i)} q_1^{(i)} q_2^{(i)}} 
		= p_1^{(i)} + p_2^{(i)} -p_1^{(i)}   p_2^{(i)}  - \rho^{(i)}   \sqrt{p_1^{(i)} p_2^{(i)} q_1^{(i)} q_2^{(i)}}
	\end{eqnarray}
	and taking into account $p_k^{(1)} = \delta_k + p_k^{(0)}$ and $p_k^{(1)}=p_k^{(0)}\mathrm{R}_1$.
	
	(iii) Replacing the probabilities of the composite endpoint with its expression in terms of the marginal parameters plus the correlation \eqref{CE_prob}, we have:
	\begin{eqnarray*} 
		\mathrm{OR}_* &=& \left(  
		\frac{1- q_1^{(1)} q_2^{(1)} - \rho^{(1)} \sqrt{ \frac{ p_1^{(1)} p_2^{(1)} }{ q_1^{(1)} q_2^{(1)} }  }  
		}{ q_1^{(1)} q_2^{(1)} + \rho^{(1)} \sqrt{ \frac{ p_1^{(1)} p_2^{(1)} }{ q_1^{(1)} q_2^{(1)} }}   
		}
		\right) \cdot \left( 
		\frac{1- q_1^{(0)} q_2^{(0)} - \rho^{(0)} \sqrt{ \frac{ p_1^{(0)} p_2^{(0)} }{ q_1^{(0)} q_2^{(0)} }  }  
		}{ q_1^{(0)} q_2^{(0)} + \rho^{(0)} \sqrt{ \frac{ p_1^{(0)} p_2^{(0)} }{ q_1^{(0)} q_2^{(0)} }}  
		} 
		\right)^{-1}  = 
		\left(   
		\frac{\frac{1}{q_1^{(1)} q_2^{(1)} } -1 - \rho^{(1)} \sqrt{ \frac{ p_1^{(1)} p_2^{(1)} }{ q_1^{(1)} q_2^{(1)} }  }}{1+ \rho^{(1)} \sqrt{ \frac{ p_1^{(1)} p_2^{(1)} }{ q_1^{(1)} q_2^{(1)} }  } }
		\right)  \left(  
		\frac{\frac{1}{q_1^{(0)} q_2^{(0)} } -1 - \rho^{(0)} \sqrt{ \frac{ p_1^{(0)} p_2^{(0)} }{ q_1^{(0)} q_2^{(0)} }  }}{1+ \rho^{(0)} \sqrt{ \frac{ p_1^{(0)} p_2^{(0)} }{ q_1^{(0)} q_2^{(0)} }  } }
		\right)^{-1} 
	\end{eqnarray*}
	Notice that:
	\begin{eqnarray*} \nonumber 
		\frac{1}{ q_1^{(i)}  q_2^{(i)}} &=& (1+ O_1^{(i)}) (1+ O_2^{(i)}) , 	  \hspace{5mm}	
		\frac{1}{ q_1^{(1)} q_2^{(1)} }  \ = \  ( 1 + \mathrm{OR}_1 O_1^{(0)}) ( 1 + \mathrm{OR}_2 O_2^{(0)})
		, 	  \hspace{8mm}	 
		\frac{ p_1^{(1)} p_2^{(1)} }{q_1^{(1)} q_2^{(1)}} =
		\ \ \mathrm{OR}_1 \mathrm{OR}_2  O_1^{(0)} O_2^{(0)} 
	\end{eqnarray*}
	Hence:
	\begin{eqnarray*} 
		\mathrm{OR}_* &=& 
		\left( \frac{ ( 1 + \mathrm{OR}_1 O_1^{(0)}) ( 1 + \mathrm{OR}_2 O_2^{(0)}) -1 - \rho^{(1)} \sqrt{ \mathrm{OR}_1 \mathrm{OR}_2  O_1^{(0)} O_2^{(0)}  }}{1+ \rho^{(1)} \sqrt{ \mathrm{OR}_1 \mathrm{OR}_2  O_1^{(0)} O_2^{(0)}  }}\right) 
		\cdot
		\left( \frac{(1+ O_1^{(0)}) (1+ O_2^{(0)}) -1 - \rho^{(0)} \sqrt{ O_1^{(0)} O_2^{(0)}  }}{1+ \rho^{(0)} \sqrt{ O_1^{(0)} O_2^{(0)} }} \right)^{-1}  
	\end{eqnarray*}
	By replacing $O_k^{(0)}$ by $\frac{p_k^{(0)}}{1-p_k^{(0)}}$, we obtain \eqref{eqOR}. 
\end{proof}

\begin{theorem}[Risk difference performance]\label{thm_diff}
	Assume that $p_k^{(0)}<1/2$ and $\delta_k<0$ ($k=1,2$).
	We denote by $\delta_*(\rho, \theta, \lambda)$ the risk difference for the composite endpoint function described in 
	\eqref{diff.prop}, specifically
	in terms of the vector of event rates $\theta$, the marginal effects $\lambda$ and the correlation $\rho$.
	Then, the risk difference for the composite endpoint for a given $\theta$ and $\lambda$  is  an increasing  function with respect to  $\rho$.
\end{theorem}

\begin{proof}[Proof of Theorem~\ref{thm_diff}]
	Observe that the difference in proportions \eqref{diff.prop} can be written as:  $$\delta_*(\rho, \theta, \lambda) = x(\theta, \lambda) + \rho \cdot y(\theta, \lambda).$$
	where: $x(\theta, \lambda) =\delta_1 q_2^{(0)} + \delta_2 q_1^{(0)}  - \delta_1 \delta_2$, and $y(\theta, \lambda) = \sqrt{p_1^{(0)}p_2^{(0)}q_1^{(0)}q_2^{(0)}} - \sqrt{p_1^{(1)}p_2^{(1)}q_1^{(1)}q_2^{(1)}}$.
	Then:
	$ \delta_*(\rho + \epsilon; \theta) - \delta_*(\rho; \theta)  = \epsilon \cdot  y(\theta, \lambda)$.
	Therefore, $ \delta_*(\rho; \theta)$ is an increasing function if and only if $ y(\theta, \lambda)>0$, $\forall \lambda, \theta$, which  is equivalent to showing that:
	\begin{eqnarray*}
		\frac{p_1^{(0)}p_2^{(0)}q_1^{(0)}q_2^{(0)}}{p_1^{(1)}p_2^{(1)}q_1^{(1)}q_2^{(1)}} >1
	\end{eqnarray*}
	It is enough to prove that for $k=1,2$,
	$$
	\frac{p_k^{(1)} q_k^{(1)} }{p_k^{(0)} q_k^{(0)} } <1
	$$
	To ease the notation call $p_k^{(1)}=a$ and $p_k^{(0)}=b$; and, by assuming $a<b<1/2$, that implies $a-b<0$ and $a+b<1$. We need to prove that:
	\begin{eqnarray*}
		\frac{a(1-a)}{b(1-b)} = \frac{a-a^2}{b-b^2} <1 \Leftrightarrow  b-a<b^2-a^2 =( b+ a)(b-a)
	\end{eqnarray*}
	Since $b-a>0$ and $a+b<1$, then we have $(a+b)(b-a)<(b-a)$. As a consequence $y(\theta, \lambda) >0$ and the risk difference of the composite endpoint is an increasing function with respect to the correlation. 
\end{proof}

\section{Derivation of the  sample size for the composite binary endpoint} \label{app.samplesizeCBE}

We establish the sample size formulae for the composite endpoint in terms of the margins and derive its properties, as outlined in Sections \ref{Section.ConvenientSampleSize} and    \ref{Section.extension}.

\normalsize
\subsection{Sample size performance according to the correlation}

\begin{lemma} \label{lemmaSS}
	Let $N(p,d) $ denote
	the sample size function for testing the difference in proportions under the unpooled variance estimate, where $p$ denotes the probability under the control group and $d$ the  relevant difference to be detected, that is:
	\begin{eqnarray}  \label{ss_unpooled} 
		N(p,d)  &=&   \left( \frac{z_\alpha + z_\beta}{d} \right)^2 \cdot \left( p \cdot (1-p) + (p+d) \cdot (1-p-d) \right) 	
	\end{eqnarray} 
	It follows that $N(p,d) $
	is an increasing function with respect to $p$ and with respect to $d$.
\end{lemma}

\begin{proof} 
	Observe that: 
	$$
	{\frac {\partial }{\partial {  p}}} N \left( {\it p},d \right) = {\frac { \left( {\it z_\alpha}+{\it z_\beta} \right) ^{2} \left( 2-4\,{\it p}-2\,d \right) }{{d}^{2}}} 
	$$
	Assuming $p < 0.5$, then $1-2p > 0$ and $2-4p0-2d > 0$. Therefore ${\frac {\partial }{\partial {  p}}} N \left( {\it p},d \right)>0$, the sample size is increasing with respect to $p$.
	Moreover,
	$$
	\frac {\partial }{\partial d}N(p,d)=
	-2\,{\frac { \left( z_\alpha+z_\beta \right) ^{2} \left(p \left( 1-p \right) + \left( d+p \right)  \left( 1-p-d \right)  \right) \\
			\mbox{}}{{d}^{3}}}+{\frac { \left(z_\alpha+z_\beta \right) ^{2} \left( 1-2p-2\,d \right) }{{d}^{2}}}
	$$
	Note that $1-2p-2d > 0$
	and therefore, $\frac {\partial }{\partial d}N(p,d)>0$; thus,   the sample size is increasing with respect to $d$.
\end{proof} 

\noindent
\textbf{Theorem 1.} Let $\theta$ and $\lambda$ be the vectors of, respectively, marginal event rates and effect sizes for the composite components, and we denote by $ \rho $ the correlation between both components.
Then,  the sample size $n(\theta, \lambda, \rho)$, for a given $\theta$ and $\lambda$ is an increasing function of the correlation $\rho$.

\begin{proof}[Proof of Theorem 1]

	Since the probability of observing the composite event is given by $\theta$ and $\rho$ (see equation \eqref{CE_prob}), and the risk difference for the composite endpoint is given by $\lambda$, $\theta$ and $\rho$ (see equation \eqref{diff.prop}), then the sample size for the composite endpoint computed by $n(\theta, \lambda, \rho)= N\left(  p_*(\theta, \rho), \delta_*(\lambda, \theta, \rho) \right)  $  is a function of $\lambda$, $\theta$ and $\rho$.
	
	To prove that the sample size for the composite endpoint  $N\left(  p_*(\theta, \rho), \delta_*(\lambda, \theta, \rho) \right)  $   increases   with $\rho$, we will show that:
	\begin{eqnarray} \label{partialN}
		\frac{\partial N( p_*(\theta, \rho), \delta_*(\lambda, \theta, \rho)) }{\partial \rho} &=&
		\frac{\partial N( p_*(\theta, \rho), \delta_*(\lambda, \theta, \rho)) }{\partial p_*(\theta, \rho)}\cdot \frac{\partial p_*(\theta, \rho)}{\partial \rho}
		+
		\frac{\partial N( p_*(\theta, \rho), \delta_*(\lambda, \theta, \rho)) }{\partial  \delta_*(\lambda, \theta, \rho)}
		\cdot \frac{\partial  \delta_*(\lambda, \theta, \rho) }{\partial \rho} >0
	\end{eqnarray}
	
	From now on we will omit $\theta$, $\lambda$ and $\rho$ and use $p_*$ and $\delta_*$ instead of   $p_*(\theta,\rho)$ and $\delta_*(\lambda, \theta, \rho)$.
	From Lemma \ref{lemmaSS}, the sample size $N(p,d)$ in \eqref{ss_unpooled} is increasing with respect to the treatment effect, $d$, and with respect to the probability of observing the event under control group, $p$, hence:
	\begin{eqnarray*}
		\frac{\partial N( p_*, \delta_*) }{\partial p_* } \ > \  0 \hspace{4mm} \text{and} \hspace{4mm} 
		\frac{\partial N( p_* , \delta_* ) }{\partial \delta_* } \ > \ 0  
	\end{eqnarray*}
	We denote by: 
	\begin{eqnarray*} 
		\hspace{-25mm}
		\frac{\partial p_*(\theta, \rho)}{\partial \rho} \ = \ -a   
		\hspace{4mm} \text{and} \hspace{4mm} 
		\frac{\partial  \delta_*(\rho) }{\partial \rho}  \ = \ a -b  
	\end{eqnarray*}
	where  $a, b>0$ and, from Theorem \ref{thm_diff}, $a-b>0$.
	Then we have:
	\begin{equation*} 
		\begin{aligned}
			\frac{\partial N( p_*, \delta_*) }{\partial \rho} &= \left(a-b\right)  
			\cdot \left( -2\,{\frac { \left( {\it z_\alpha}+{\it z_\beta} \right) ^{2} \left( p_* \, \left( 1-p_*(\theta,\rho) \right)  
					+ \left( \delta_*+p_*  \right)  \left( 1-p_* -\delta_* \right)  \right)  }{{\delta_*}^{3}}} \right.
			\\
			&\left.+{\frac { \left( {\it z_\alpha}+{\it z_\beta} \right) ^{2} \left( 1-2\,p_* -2\delta_* \right) }{{\delta_*}^{2}}} \right) 
			-a   \cdot\frac { \left( {\it z_\alpha}+{\it z_\beta} \right) ^{2} \left( 2-4\,p_* -2\delta_* \right)  }{{\delta_*}^{2}}  
		\end{aligned}
	\end{equation*}	
	and this is positive if and only if:
	
	\begin{equation} \label{eq_larger1} 
		\begin{aligned}
			&\left(a-b\right)
			\cdot \left( -2 {\frac {p_*\, \left( 1-p_* \right) + \left( \delta_*+p_* \right)  \left( 1-p_*-\delta_* \right) }{\delta_*}}+1-2 p_*-2\delta_* \right) - \left( 2-4\,p_*-2\delta_* \right) \cdot  a 
		\end{aligned}
	\end{equation}
	\eqref{eq_larger1} is positive.	Then we have:
	\begin{equation}  \label{eq_larger2} 
		\begin{aligned}
			& - 2\frac{\left(a-b \right)}{\delta_*}    \left(  p_*  \left( 1-p_* \right) +  \left( \delta_*+p_* \right)  \left( 1-p_*-d \right)  \right)  - b \left( 1-2\,p_*-2\delta_* \right)  
			\mbox{}+ a \left( -1+2\,p_* \right)\\
			& > 
			- 2\frac{\left(a-b \right)}{\delta_*}    p_*  \left( 1-p_* \right) - 2\frac{\left(a-b \right)}{\delta_*}   \left( \delta_*+p_* \right)  \left( 1-p_*-\delta_*\right)     - 2 a \left( 1 - p_*- \delta_* \right)  +2ap_*
		\end{aligned}
	\end{equation}
	Then \eqref{eq_larger1}  > \eqref{eq_larger2}, because  $a>b$.
	Note that the first and forth terms are positive, so we end if we see that the second plus third are also positive. This follows from the fact that:
	\begin{eqnarray*}
		- 2\frac{\left(a-b \right)}{\delta_*}   \left( \delta_*+p_* \right)    - 2 a    >  0 &\Leftrightarrow &
		a\left( 1+ \frac{\delta_*+p_*}{\delta_*}\right)  <  b\left(   \frac{\delta_*+p_*}{\delta_*}\right)		 \Leftrightarrow 
		\frac{a}{b} > \frac{    \frac{\delta_*+p_*}{\delta_*}}{1+   \frac{\delta_*+p_*}{\delta_*}}
	\end{eqnarray*}
	Since $a,b>0$ and $a-b>0$, we have $ \frac{a}{b}>1$; and since 
	$\left( \frac{\delta_*+p_*}{\delta_*} \right)$,\ 
	$\left(  1+   \frac{\delta_*+p_*}{\delta_*}\right)<0$, we have 
	$\left( \frac{\delta_*+p_*}{\delta_*} \right) \big/  \left(  1+   \frac{\delta_*+p_*}{\delta_*}\right)  \in (0,1)$. Therefore \eqref{partialN} is positive, as we intended to prove.
\end{proof}

\newpage


\clearpage

%

\end{document}